\documentclass[conference]{IEEEtran}
\IEEEoverridecommandlockouts
\usepackage{cite}
\usepackage{amsmath,amssymb,amsfonts}
\usepackage{amsthm} 
\usepackage{algorithmic}
\usepackage{graphicx}
\usepackage{cancel}
\usepackage{float}
\usepackage{textcomp}
\usepackage{xcolor}
\usepackage{times}

\usepackage{soul}
\usepackage{url}
\usepackage[utf8]{inputenc}
\usepackage[small]{caption}

\usepackage{amsmath}
\usepackage{booktabs}
\usepackage{comment}
\usepackage{amsthm}
\newtheorem{defi}{Definition}
\newtheorem{theorem}{Theorem}

\newtheorem{lemma}{Lemma}
\usepackage{bm}

\usepackage{graphicx}
\usepackage{caption}
\usepackage{subcaption}
\usepackage{hyperref}
\usepackage[subrefformat=parens,skip=0pt]{subcaption}

\usepackage[linesnumbered,ruled,vlined]
{algorithm2e}
\def\BibTeX{{\rm B\kern-.05em{\sc i\kern-.025em b}\kern-.08em
    T\kern-.1667em\lower.7ex\hbox{E}\kern-.125emX}}
\begin{document}

\title{Deadline-Aware Online Scheduling for LLM Fine-Tuning with Spot Market Predictions}

%

\author{
\IEEEauthorblockN{Linggao Kong\IEEEauthorrefmark{1}, Yuedong Xu\IEEEauthorrefmark{1}, Lei Jiao\IEEEauthorrefmark{2}, Chuan Xu\IEEEauthorrefmark{3}}
\IEEEauthorblockA{\IEEEauthorrefmark{1}Fudan University, China \quad
\IEEEauthorrefmark{2}University of Oregon, USA \quad
\IEEEauthorrefmark{3}Inria, France \\
}
\thanks{Yuedong Xu is with College of Computer Science and Artificial Intelligence, and Artificial Intelligence Innovation and Incubation Institute, Fudan University, Shanghai, China (e-mail: ydxu@fudan.edu.cn)}
}

\maketitle

\begin{abstract}

As foundation models grow in size, fine-tuning them becomes increasingly expensive. While GPU spot instances offer a low-cost alternative to on-demand resources, their volatile prices and availability make deadline-aware scheduling particularly challenging. We tackle this difficulty by using a mix of spot and on-demand instances. 
Distinctively, we show the predictability of prices and availability in a spot instance market, the power of prediction in enabling cost-efficient scheduling and its sensitivity to estimation errors. 
An integer programming problem is formulated to capture the use of mixed instances under both the price and availability dynamics. We propose an online allocation algorithm with prediction based on the committed horizon control approach that leverages a \emph{commitment level} to enforce the partial sequence of decisions. 
When this prediction becomes inaccurate, we further present a complementary online algorithm without predictions. 
An online policy selection algorithm is developed that learns the best policy from a pool constructed by varying the parameters of both algorithms. 
We prove that the prediction-based algorithm achieves tighter performance bounds as prediction error decreases, while the policy selection algorithm possesses a regret bound of $\mathcal{O}(\sqrt{T})$. 
Experimental results demonstrate that our online framework can adaptively select the best policy under varying spot market dynamics and prediction quality, 
consistently outperforming baselines and improving utility by up to 54.8\%.

\end{abstract}

\section{Introduction}

Pretrained large language models (LLMs) such as GPT~\cite{achiam2023gpt} and Llama~\cite{touvron2023llama} have deeply penetrated into our daily lives. 
Although these models demonstrate broad linguistic understanding, factual knowledge, and reasoning, they often lack the precision and task-specific performance required in real-world applications, making fine-tuning with domain-specific data essential to adapt general-purpose LLMs. A variety of fine-tuning methods have been developed, including full fine-tuning, adapter-based fine-tuning, prompt tuning, instruction tuning, and reinforcement learning from human feedback (RLHF)~\cite{howard2018universal,houlsby2019parameter,lester2021power,wei2021finetuned,christiano2017deep}. 
Among these, LoRA (Low-Rank Adaptation)~\cite{hu2022lora} is notable for its efficiency and scalability. It updates a small number of trainable parameters while keeping the base model frozen, significantly reducing memory and computation costs, which makes it particularly well-suited for rapid adaptation, cost-effective deployment, and efficient operation even in resource-constrained environments.

Finetuning LLMs with LoRA typically requires high-end cloud GPUs, incurring substantial costs. Recently, spot GPU instances have emerged as a cost-effective alternative, offering up to 91\% discounts~\cite{gcpSpotDoc}. Third-party brokers further aggregate GPUs from decentralized providers into federated clouds~\cite{vastai2025,runpod2025,awsLambda}, where the available GPUs are de facto spot instances. However, spot instances exhibit \emph{high price volatility} and \emph{intermittent availability} due to preemption and provider churn, introducing a trade-off between cost savings and meeting service-level objectives (SLOs).

CPU spot instance usage in cloud computing has been widely studied over the past decade. For instance, Menache et al.~\cite{menache2014demand} proposed an online algorithm to determine both the type (on-demand vs. spot) and quantity of instances to use at any given time. Song et al.~\cite{song2012optimal} developed a dynamic bidding strategy from the perspective of a cloud broker, aiming to maximize service profit by bidding for spot instances. In contrast, the use of GPU spot instances in deep learning systems has only recently gained attention. Yang et al.~\cite{mao2025skyserve} formulated an integer programming problem to assign CPU or GPU spot instances to multi-tenant jobs, which they solved using a polynomial-time algorithm based on linear programming rounding. Wu et al.~\cite{wu2024can} studied the strategy for switching between spot and on-demand GPU instances to minimize costs, considering the inherent variability of spot availability.

In this paper, we aim to minimize the cost of LoRA-based LLM fine-tuning with deadline awareness in a hybrid market consisting of both on-demand and spot GPU instances. Our problem is characterized by three key features. First, we explore the impact of dynamically adjusting the number of GPU instances on model convergence—an aspect that has not been previously validated. Second, unlike prior work that assumes the unpredictability of spot instance prices and availability, we leverage their potential predictability to design more effective scheduling strategies instead. Third, we significantly expand the state space by incorporating dynamic spot instance prices and GPU availability, in contrast to prior work that assumes fixed price and binary availability indicators for spot instances.

The uncertainty in prediction quality and the enlarged decision space from fluctuating spot prices and availability complicate resource allocation and scheduling. We first formulate scheduling as a mixed-integer optimization problem. To maximizng the objective under volatile spot market conditions, we propose a prediction-based allocation algorithm building on Committed Horizon Control algorithm and a fallback heuristic leveraging real-time job progress and market trends. Subsequently, we develop an online policy selection algorithm that adaptively identifies the best policy from a pool constructed by systematically varying the hyperparameters of these two algorithms. We rigorously prove that smaller prediction errors yield a tighter upper bound on the performance gap of our prediction-based allocation algorithm relative to the optimal policy. 
Furthermore, we show that the online policy selection achieves a regret bound of $\mathcal{O}(\sqrt{T})$. Extensive experiments demonstrate that our approach consistently outperforms baselines---On-Demand Only Policy (OD-Only), Maximal Spot Utilization Policy (MSU), and Uniform Progress Policy (UP)~\cite{wu2024can}---and adapts effectively across diverse prediction environments.

This work makes the following contributions:

\subsubsection{Deadline-Aware Fine-Tuning Problem}

We formulate a deadline-aware resource allocation problem for fine-tuning jobs on cloud platforms using mixed on-demand and spot instances to minimize cost.

\subsubsection{Hybrid Resource Allocation with Online Policy Selection}

We design a prediction-based allocation algorithm with performance guarantees that improve with prediction accuracy, along with a non-predictive variant for robustness. To handle dynamics, we propose an online policy selector that adaptively identifies the best strategy from a hyperparameterized pool, with $\mathcal{O}(\sqrt{T})$ regret.

\subsubsection{Experimental Evaluation}

We evaluate our approach under different prediction errors and dynamic conditions. Results show that our algorithms adapt effectively and consistently outperforms baselines. In typical
configurations, the best-selected policy by our adaptive algorithm improves utility by 49.0\%, 54.8\%, and 33.4\% over the three baselines.

\section{Motivation}
With the scaling of foundation models, their fine-tuning for downstream tasks becomes costly on cloud platforms. This section highlights the cost-time trade-off introduced by using cheaper yet unreliable spot GPUs, where spot prediction can greatly enhance scheduling decisions.

\subsection{Multi-Instance Parallel Fine-tuning}

LoRA is an important fine-tuning method for LLMs that decomposes the incremental weight matrix  \( \Delta W \) into two low-rank matrices $B$ and $A$ at each layer, i.e. $W = {W_0} + \Delta W \simeq {W_0} + B\cdot A$, where \( W_0 \) is the original LLM model weight. Therefore, the number of learnable parameters shirnks from from $d^2$ to $2d\cdot r$ with $d$ denoted as the dimension of matrix $W$ and $r$ $(r\ll d)$ denoted as the low dimensional rank. Even though, finetuning LLMs is still compute intensive, which usually demands multi-GPU data parallelism. We evaluate the impact of the number of NVIDIA A100 GPUs on the training throughput in Figure \ref{fig3:throughput}. Both ChatGLM3-6B and Llama2-7B are considered, and the batch size is unanimously set to 32. Here, $x$-coordinate denotes the number of GPUs used in a machine and $y$-coordinate shows the fine-tuning throughput in samples per second. One can observe that the training throughput increases almost linearly with the number of GPUs, despite the models used. 

LoRA fine-tuning with spot instances is feasible on both cloud platforms and federated edge clouds. For example, consider LLaMA2-7B, which has a model dimension of 4096, 32 layers, and a LoRA rank of 16. When the low-rank matrices are updated in each iteration, the total communication overhead amounts to approximately 16.8MB per iteration for half-precision gradients. If two spot instance GPUs are hosted on different machines but connected via a 200Gbps RDMA link, the communication time per GPU is negligible compared to the computation time (e.g., 10 seconds). However, if they are connected over standard Internet links with a bandwidth of 100Mbps, the communication time per iteration is approximately 1.35 seconds. This overhead can be effectively amortized by performing multiple rounds of local training before communications.

Dynamically adjusting the degree of data parallelism can introduce nontrivial overhead, particularly in spot instance scenarios. The reconfiguration time includes the transmission of the training checkpoint, the startup time of the Docker container, and the initialization of the fine-tuning process. 
A typical checkpoint (model + LoRA parameters + optimizer state) takes 0.58s over 200Gbps RDMA, but up to 1152s over 100Mbps links. We define this as the \emph{switching cost} during spot preemption, which depends on network speed but not GPU count, since only new instances require reconfiguration.

\begin{figure}[t]
    \centering
    \begin{minipage}[t]{0.49\linewidth}
        \centering
        \includegraphics[width=\linewidth]{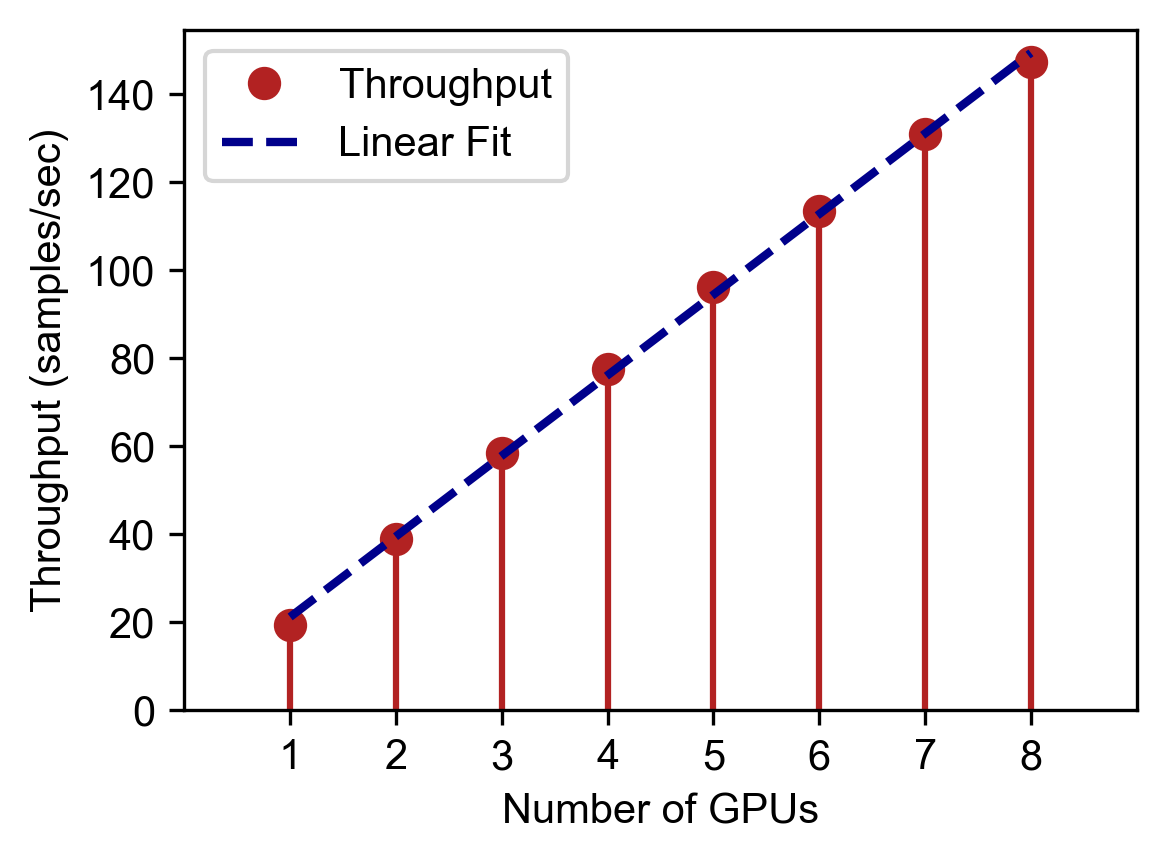}
        \subcaption{ChatGLM3-6B}
        \label{fig1a:ChatGLM3-6B}
    \end{minipage}%
    \hspace{-0.3em}  
    \begin{minipage}[t]{0.49\linewidth}
        \centering
        \includegraphics[width=\linewidth]{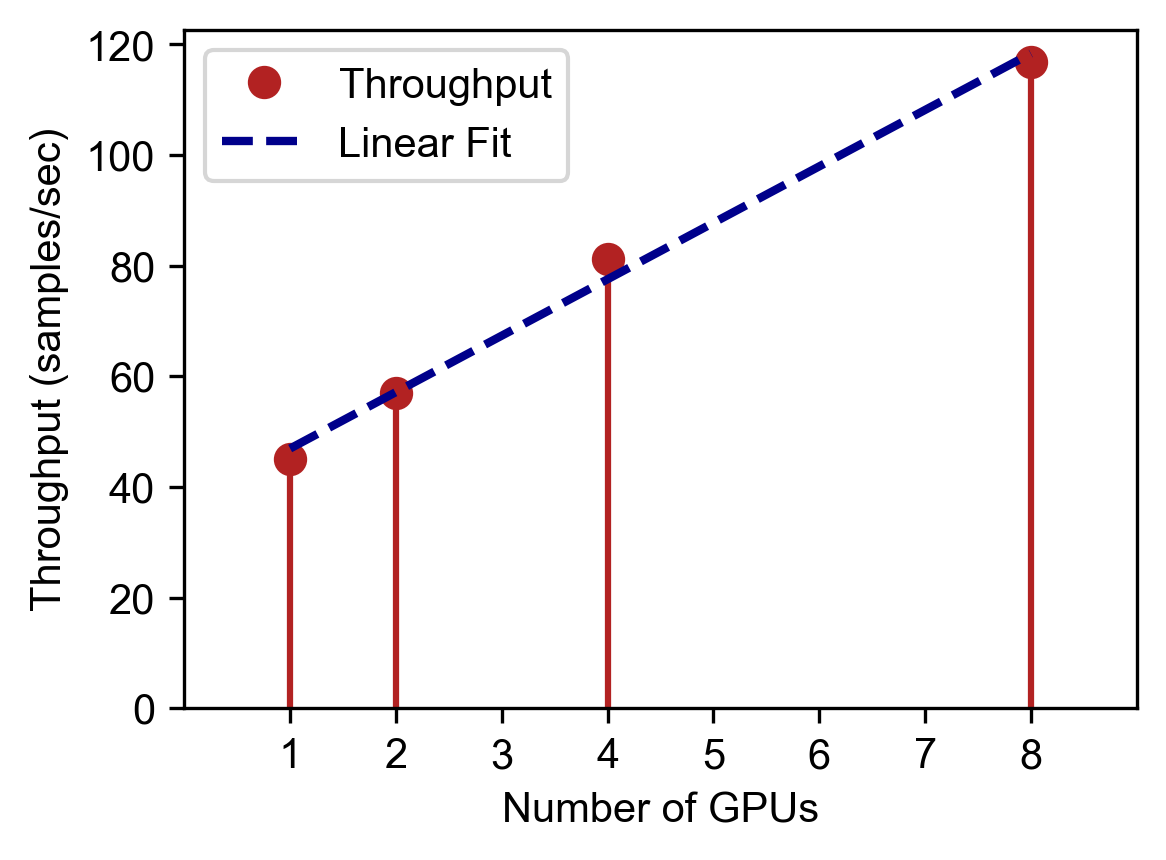}
        \subcaption{Llama2-7B}
        \label{fig1b:Llama2-7B}
    \end{minipage}
    \vspace{-0.5em}
    \caption{Training throughput vs number of GPUs (A100 GPU).}
    \label{fig3:throughput}
\end{figure}

\subsection{Volatility of Spot Instance Price and Availability}

Service providers offer two pricing models for GPU resources: on-demand instances and spot instances. On-demand instances are available whenever needed and provide guaranteed reliability, following a pay-as-you-go model. In contrast, spot instances are significantly more cost-effective but come with the trade-off of intermittent availability. Chasing the lowest spot instance prices may reduce costs but risks violating the SLO, while relying solely on on-demand instances results in high fine-tuning expenses.

Figure \ref{fig1:fluctuations} illustrates the availability and price fluctuations of NVIDIA A100 GPUs over a 10-day period on Vast.ai~\cite{vastai2025}, a cloud computing platform. We collect data at 30-minute intervals, recording the number of available and rented GPU instances, along with both the free and rented prices. As shown in Figure~\ref{fig1:fluctuations}(b), the median instance price is only about 60\% of the P90 price. 
Hence, using spot instances wisely will lead to a significant cost saving in fine-tuning. In Figure~\ref{fig1:fluctuations}(a), the number of available GPUs on the platform fluctuates over time. Since Vast.ai aggregates all the GPUs across its clients, the number of GPUs available within a specific geographic region, especially those suitable for LLM fine-tuning, is often limited, making spot instance availability unreliable. Unlike the existing study~\cite{wu2024can} that consider only the dynamics of spot GPU instance availability, our work accounts for both price and availability fluctuations. This motivates us to develop an online learning algorithm for purchasing on-demand or spot instances so as to minimize the cost with SLO guarantee.

\begin{figure}[t]
    \centering
    \begin{minipage}[t]{0.49\linewidth}
        \centering
        \includegraphics[width=\linewidth]{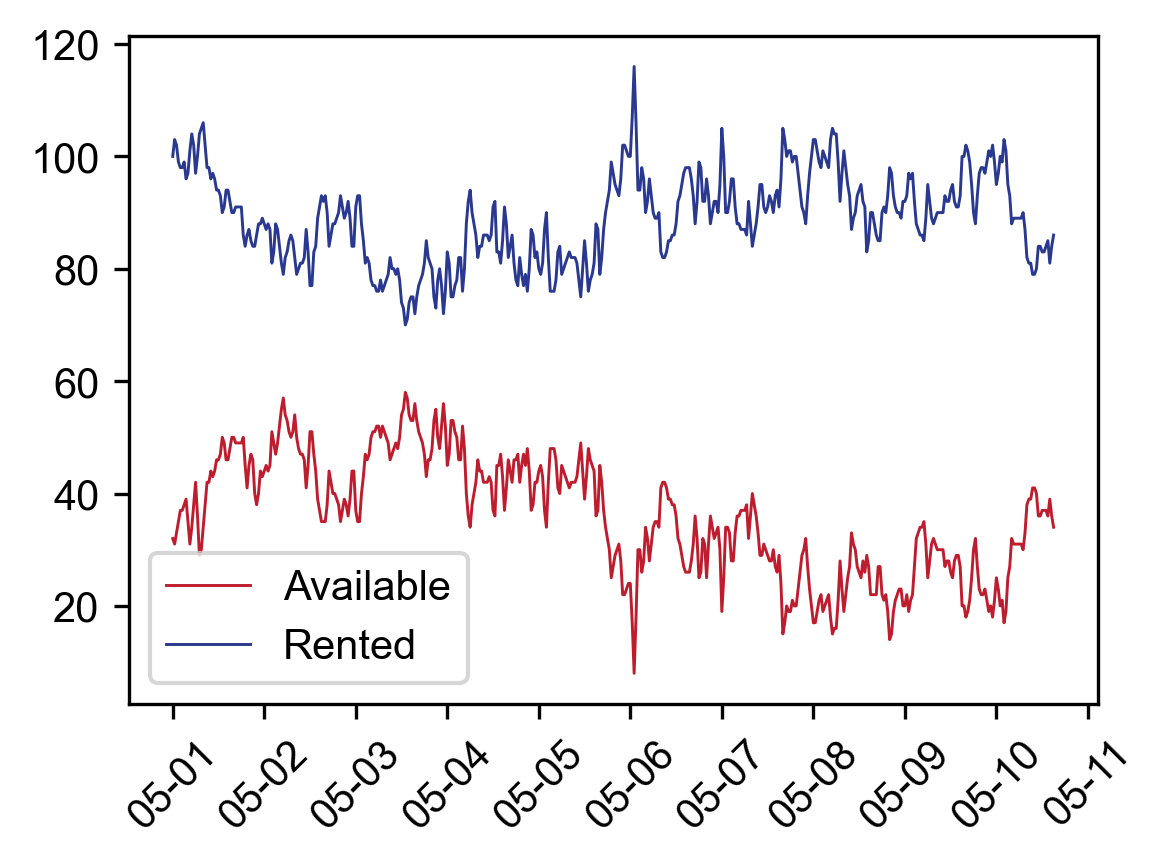}
        \subcaption{Availability}
        \label{fig2a:avail fluctuations}
    \end{minipage}%
    \hspace{-0.3em}  
    \begin{minipage}[t]{0.49\linewidth}
        \centering
        \includegraphics[width=\linewidth]{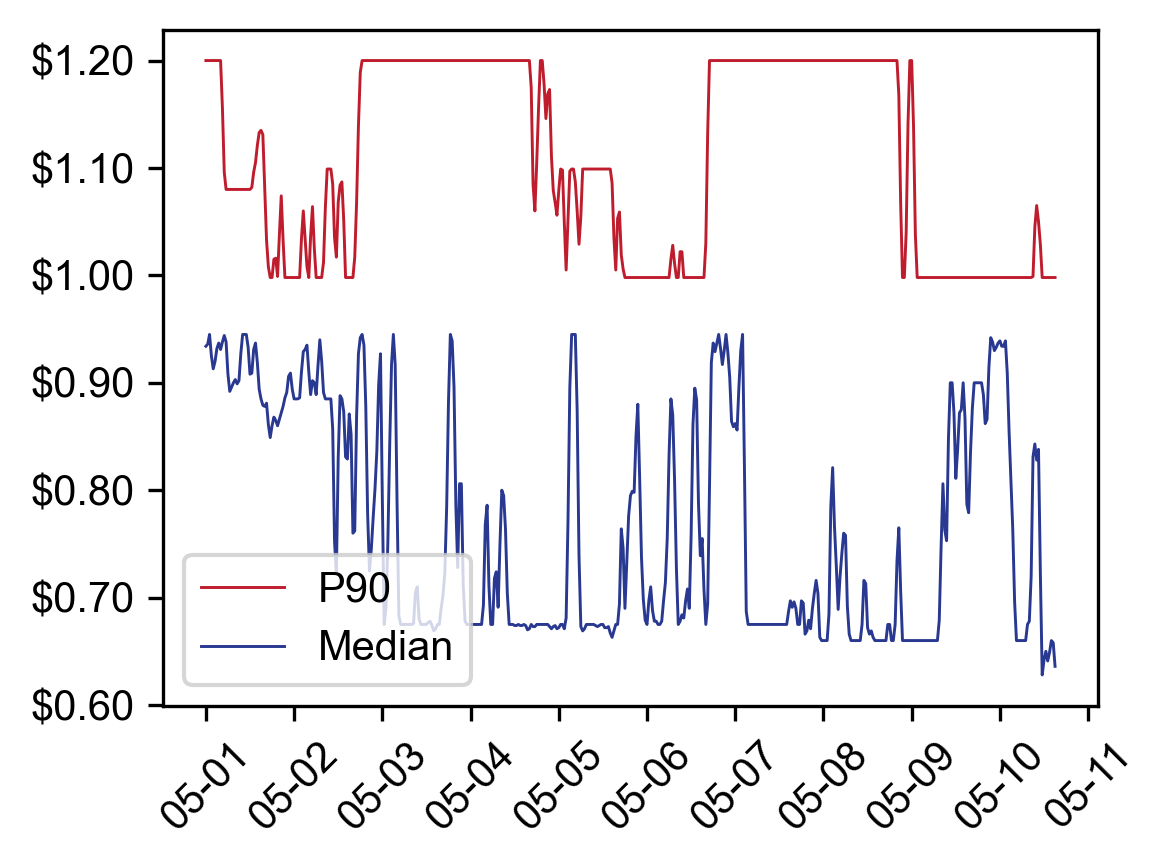}
        \subcaption{Price}
        \label{fig2b:price fluctuations}
    \end{minipage}
    \vspace{-0.5em}
    \caption{Fluctuations in A100 Spot Instances on Vast.ai over 10 Days.}
    \label{fig1:fluctuations}
\end{figure}

\subsection{Understanding the Power of Prediction}


The price and availability dynamics are not completely random, but is predictable to a certain extent. Under this circumstance, an online algorithm is designed to minimize regret, defined as the difference between its performance and the best possible fixed strategy in hindsight. By anticipating the future inputs, online learning with prediction is likely to make more informed choices, thus reducing the expected cumulative regret. Unfortunately, the power of prediction has been overlooked in online resource allocation for networking applications. 

Preliminary analysis of Figure \ref{fig1:fluctuations} indicates that spot instance availability, such as that of A100 GPUs, tends to follow a daily trend, with higher availability during the daytime than at night. We employ an Auto-Regressive Integrated Moving Average (ARIMA) model~\cite{box2015time} to learn the temporal dynamics of instance availability and pricing, using a 30-minute time window for prediction. As illustrated in Figure \ref{fig2:Forecasting}, our predictions closely match the actual fluctuations. When predictions are accurate, this substantially reduces uncertainty and enables more aggressive allocation of low-cost spot resources for training jobs, thereby reducing overall costs. 

\begin{figure}[t]
    \centering
    \begin{minipage}[t]{0.49\linewidth}
        \centering
        \includegraphics[width=\linewidth]{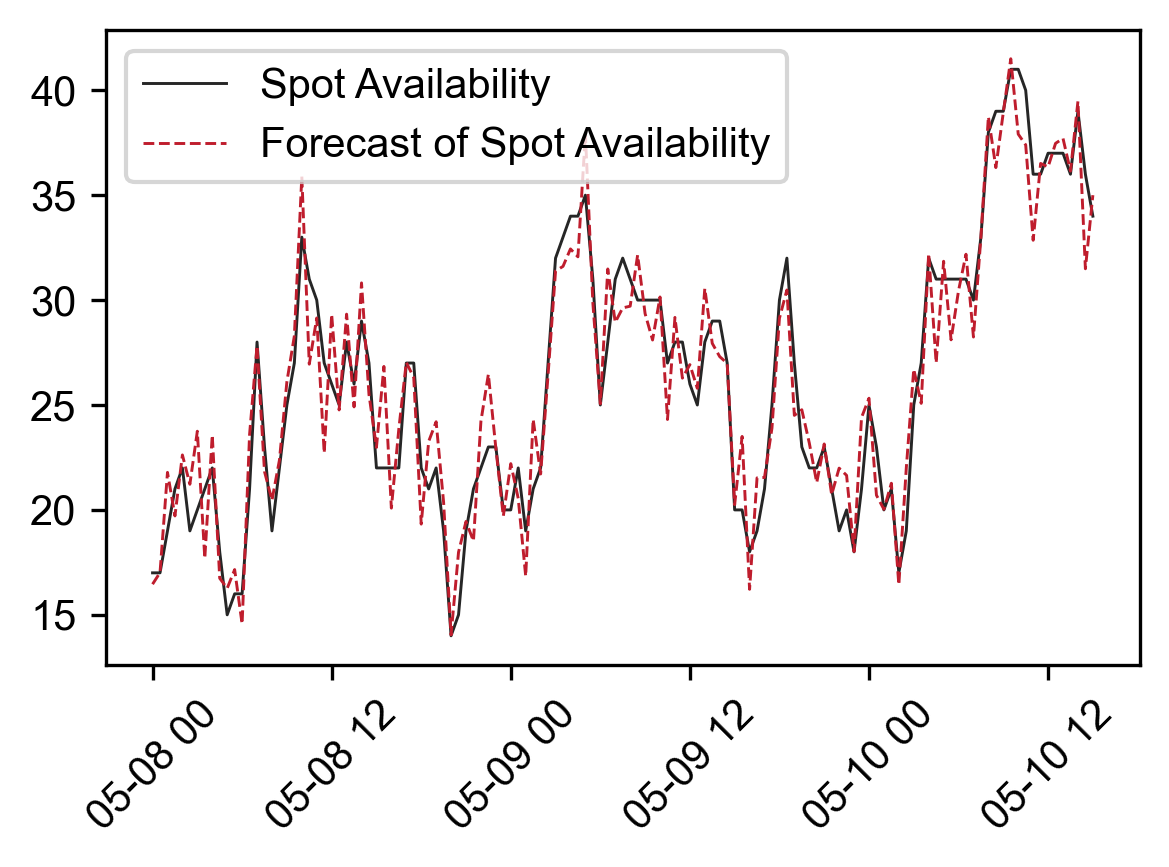}
        \subcaption{Availability}
        \label{fig2a:ARIMA_avail_1}
    \end{minipage}%
    \hspace{-0.3em}  
    \begin{minipage}[t]{0.49\linewidth}
        \centering
        \includegraphics[width=\linewidth]{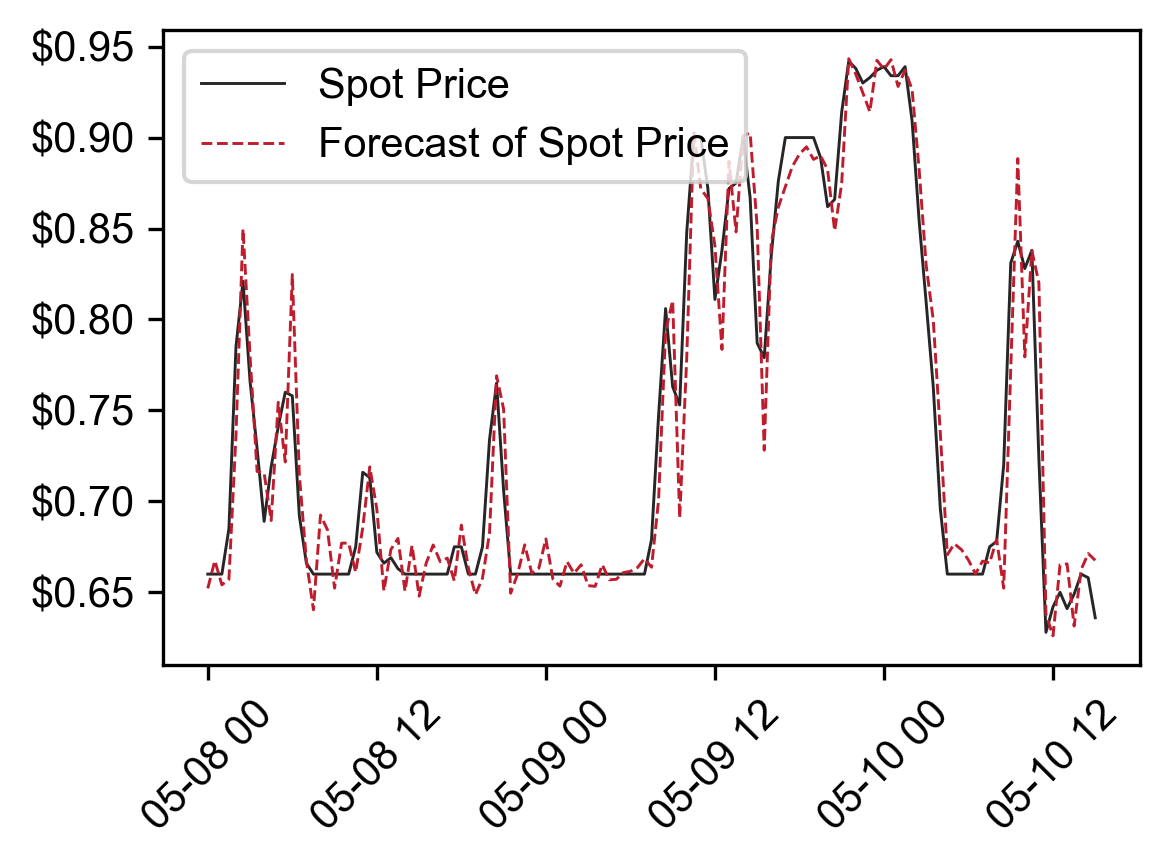}
        \subcaption{Price}
        \label{fig2b:ARIMA_price_1}
    \end{minipage}
    \vspace{-0.5em}
    \caption{Forecasting Spot Availability and Price Using ARIMA.}
    \label{fig2:Forecasting}
\end{figure}

However, due to platform-specific pricing and the unpredictability of real-world events, predictions may be unreliable. Different machine types and regions potentially exhibit varying fluctuation patterns, leading to low accuracy or unpredictable behavior. In such cases, relying on forecasts could increase costs and risks, making the consequences of inaccurate predictions critical to consider.

\begin{figure}[t] 
\centering
\includegraphics[width=\columnwidth]{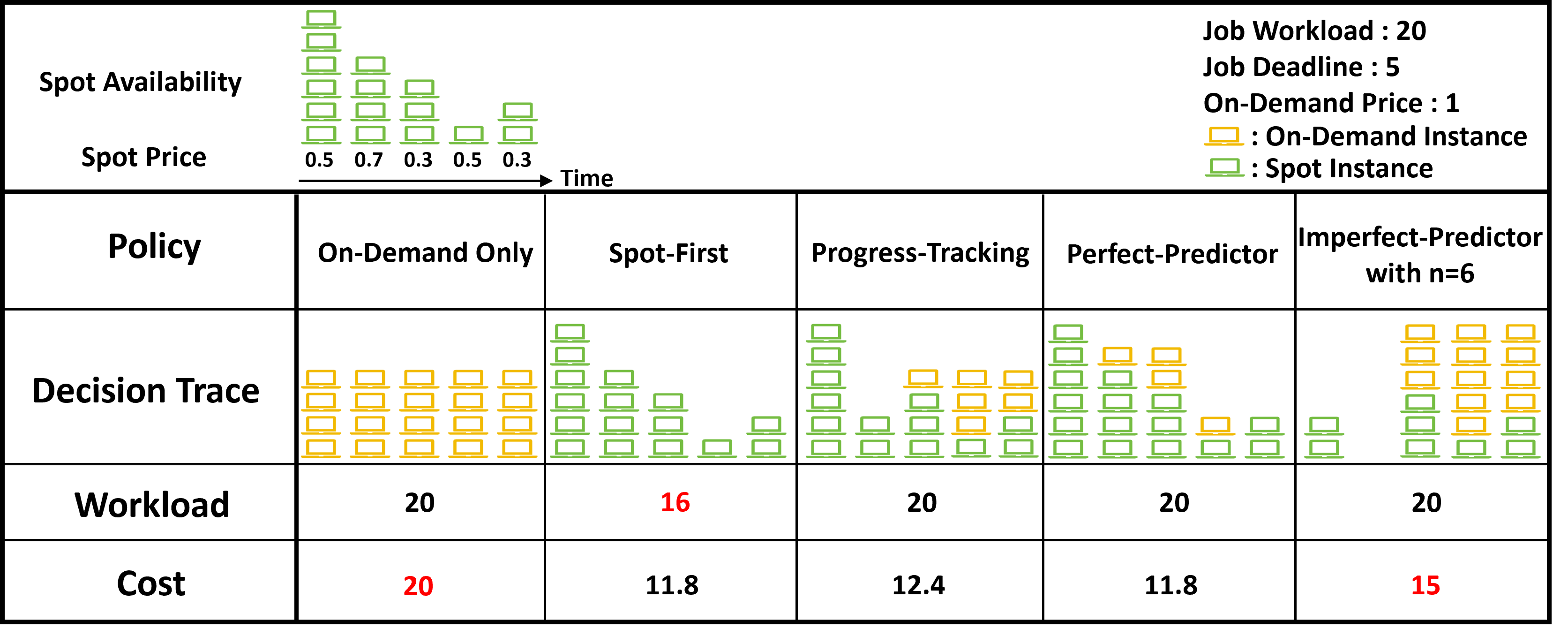} 
\caption{Comparison of Workload and Cost under Different Resource Allocation Strategies.}
\label{fig4:toyexample}
\end{figure}

Figure~\ref{fig4:toyexample} presents a comparative analysis of resource allocation strategies with and without prediction. We simulate a fine-tuning job that must complete 20 units of workload within 5 time slots, ignoring reconfiguration overhead for simplicity. Baseline strategies without prediction include: (1) using only on-demand instances, (2) prioritizing spot instances, and (3) a hybrid progress-tracking strategy that adjusts allocations based on partial progress. Prediction-based strategies include (i) perfect foresight and (ii) a constant forecast of 6 available spot instances, both of which incorporate on-demand instances as fallback resources to ensure deadline compliance. Results show that using only on-demand instances guarantees deadline compliance but incurs high costs; prioritizing spot instances reduces costs but may violate deadlines; the progress-tracking strategy ensures timely progress but fails to exploit cheaper spot resources effectively. In contrast, when predictions are accurate, the system can utilize low-cost spot instances while maintaining progress guarantees, resulting in the lowest overall cost. However, when predictions are inaccurate, the total cost can even exceed that of non-predictive strategies.

In summary, while accurate forecasting can substantially reduce costs by enabling better utilization of spot instances, it is crucial to develop mechanisms that can gracefully handle cases of poor or unpredictable forecast performance.

\section{System Model}

In this section, we model the allocation of hybrid instances to a fine-tuning job with a deadline constraint. Our objective is to achieve an optimal balance between SLO and cost.

\subsection{Modeling of Fine-tuning Job Processing}

We consider the problem of fine-tuning large foundation models using LoRA~\cite{hu2022lora}, where fine-tuning jobs arrive sequentially in a stochastic manner. For modeling clarity, we focus on a single job that arrives at the system, while our framework can be readily extended to handle multiple jobs. A fine-tuning job is characterized by a four-tuple \( \{ L, d, N^{\text{min}}, N^{\text{max}} \} \), where \( L \) represents the total computation workload required for completion, \( d \) is the preset completion deadline, and \( N^{\text{min}} \) and \( N^{\text{max}} \) represent the minimum and maximum number of GPUs required to process the job. Here, $L$ is computed by ${L} = {D} \times {n^{epoch}}$, where $D$ is the number of data samples in fine-tuning and $n^{epoch}$ is the number of epochs. 

The degree of GPU parallelism for the job is limited. We define $N^{min}$ as the minimum number of GPUs required to store all necessary components in their HBM for fine-tuning, including the base model, the LoRA adapter parameters, any associated data, and, if applicable, the checkpoint and optimizer states. The maximum parallelism $N^{max}$ is the highest number of GPUs that can be used without significantly diminishing parallel efficiency due to under-utilization or overhead.

\subsection{Modeling of Computing Time}

We consider using multiple instances for parallel computing to improve computational efficiency, and a mixed use of on-demand and spot instances to reduce costs. Without loss of generality, we evenly divide time into discrete time slots. 
At each time slot, the available quantity of spot instances is denoted by \( n_t^{avail} \), and their corresponding price is \( p_t^s \), which remain constant within each time slot and only change across discrete time periods. On-demand instances are considered continuously available with a fixed price \( p^o \). Any spot and on-demand instances are considered homogeneous in terms of GPU computing capability. Based on the real-world experimental results illustrated in Figure~\ref{fig3:throughput}, we represent the relationship between throughput and the number of instances as
\begin{equation}
    H(n) = \left \{
    \begin{aligned}
         &\alpha \cdot n + \beta \quad (\beta \neq 0), && n\in \mathbb{Z}^+, \\
         &0, && n=0,
    \end{aligned}
    \right.
    \label{throughput estimation}
\end{equation}
where $n$ is the number of computing instances, and $\alpha$ stands
for the improved efficiency with the increased scale of parallelism. To avoid affecting the model's convergence due to changes in the number of instances, we fix the global batch size. While we adopt a linear formulation for clarity, the framework remains applicable to nonlinear throughput models. 

Additionally, when the total number of instances changes, other overheads are introduced, such as launching new instances, data transfer, and synchronization operations between multiple instances. We denote the proportion of effective computation time within the time slot as:
\begin{equation}
    {\mu_t} = \left\{\begin{array}{l}
{\mu_1},{n_t} > {n_{t - 1}},\\
{\mu_2},{n_t} < {n_{t - 1}},\\
1,{n_t} = {n_{t - 1}},
\end{array} \right.
\end{equation}
where \( n_t \) is the total number of instances at time slot \( t \). 
We have ${\mu _1} \le {\mu _2} \le 1$, where $\mu_1$ accounts for the additional overhead introduced by both instance initialization and system reconfiguration, and $\mu_2$ reflects only the reconfiguration cost.

\subsection{Modeling of Cost and Revenue}

At each time slot \( t \), we allocate $n_{t}^o$ on-demand instances and $n_{t}^s$ spot instances for the job, until the entire job is completed at time \( T \). Therefore, the total cost of completing the job is 
\begin{equation}
C\left( {n_{t}^o,\left. {n_{t}^s} \right)} \right. = \sum\limits_{t = 1}^{{T}} {\left( {n_{t}^o \cdot {p^o} + n_{t}^s \cdot p_t^s} \right)}.
\end{equation}

We denote the revenue as a function of the deadline and the completion time~\cite{cheng2018deadline}.  Since the time requirement for fine-tuning jobs is elastic, we denote \( d \) as a soft deadline, and we also use another hard deadline $\gamma \times d$ ($\gamma > 1$) to appropriately tolerate delays in job completion. We denote the value function for completing the job as
\begin{equation}
    V\left( {{T}} \right) = \left \{
    \begin{aligned}
         &{v},&&{T} \le {d},\\
         &{v} \cdot \left( {1 - \frac{{{T} - {d}}}{{\left( {{\gamma } - 1} \right) \cdot {d}}}} \right),&&{d} < {T} < {\gamma }{d},\\
         &0,&&{T} \ge {\gamma }{d}.
    \end{aligned}
    \right.
    \label{value function t}
\end{equation}
In our framework, we do not impose constraints on the specific form of the value function. As long as the value function approximately follows the aforementioned pattern, it remains effective in yielding desirable results.

\subsection{Problem Formulation}

We formulate our GPU resource allocation problem for a fine-tuning job as follows:
\begin{equation}
\hspace{-5em} 
\max \quad V\left( T \right) - C\left( n_{t}^o, n_{t}^s \right)
\tag{5}
\label{objective}
\end{equation}

\vspace{-1.9em} 

\begin{subequations}
\begin{align}
\text{s.t.} \quad 
& \sum_{t=1}^{T} \mu_t \cdot H\left( n_{t}^o + n_{t}^s \right) \ge L, \tag{5a} \\
& n_{t}^s \le n_t^{\mathrm{avail}}, \quad \forall t, \tag{5b} \\
& n_{t}^o + n_{t}^s \le \delta_t \cdot N^{\max}, \quad \forall t, \tag{5c} \\
& n_{t}^o + n_{t}^s \ge \delta_t \cdot N^{\min}, \quad \forall t, \tag{5d} \\
& \delta_t \in \{0,1\}, \quad n_{t}^o, n_{t}^s \in \mathbb{N}, \quad \forall t. \tag{5e}
\end{align}
\end{subequations}

The objective is to maximize the profit of executing the job, calculated as the completion-time revenue minus the accumulated cost.
Constraint (5a) ensures that the job is completed at time \( T \). Constraint (5b) ensures that at each time slot \( t \), the number of spot instances does not exceed the spot availability. Constraints (5c) and (5d) define the range of the total number of instances at each time slot \( t \). When the job's state is pending, i.e., \( \delta_t = 0 \), both the number of on-demand and spot instances are $0$. When the job's state is executing, i.e., \( \delta_t = 1 \), the total number of on-demand and spot instances should be within the range $[N^{min},N^{max}]$. Constraint (5e) specifies the domain of the variables.

\subsection{Problem Reformulation}
Our goal is to design a prediction-based online algorithm to optimize the objective. 
However, the original problem cannot be directly addressed using standard prediction-based methods. These algorithms operate within a fixed prediction window in each iteration, which is misaligned with the objective (\ref{objective}) that depends on the uncertain and potentially unbounded job completion time $T$, which is influenced by future decisions. To address this, we first refine the problem formulation through several key transformations.

\subsubsection{Workload Slicing for Deadline Achievement}At each time slot, we only have access to the current and historical spot instance prices and availability, while future data remains unknown. To ensure the job is completed before the deadline, 
we predefine a reference progress trajectory by slicing the total workload $L$ over the deadline $d$.
If the corresponding portion of the workload can be completed in each time slot, the entire job will be finished by the deadline. Here, we use uniform slicing, where the expected progress of the job at each time $t$ is
\begin{equation}
Z_{t}^{\exp } = \frac{{{L}}}{{{d}}} \cdot t.
\label{expected progress}
\end{equation}

\subsubsection{Transformation of the Value Function}
\label{subsubsec:reformulation object}
The objective function consists of two components with different decision variables: the value function depends on the completion time \( T \), while the cost function depends on \( n_t^o \) and \( n_t^s \). However, since \( T \) can only be determined after the job is completed, it introduces uncertainty into the optimization process. To address this, we express \( T \) as a function of the cumulative allocated resources across time slots, allowing the entire objective to be written in terms of \( n_t^o \) and \( n_t^s \) only.

Given that the job's revenue rapidly diminishes to zero when the completion time exceeds the deadline, as described in (\ref{value function t}), we introduce a termination configuration that simply choose on-demand GPU instances with the maximum parallelism to complete the workload immediately, and different termination schemes can be adopted without changing the modeling framework. Consequently, we only need to track the completed workload by the deadline to deduce the total completion time $T$.
We denote the workload completed before the deadline as
\begin{equation}
{Z^{\text{ddl}}} = \sum\limits_{t = 1}^{{d}} {{\mu_t} \cdot {H\left( {n_{t}^o + n_{t}^s} \right) } },
\end{equation}
and the monetary cost before deadline as
\begin{equation}
C^\text{ddl}\left( {n_{t}^o,\left. {n_{t}^s} \right)} \right. = \sum\limits_{t = 1}^{{d}} {\left( {n_{t}^o \cdot {p^o} + n_{t}^s \cdot p_t^s} \right)}.
\end{equation} 
By incorporating the relationship between $T$ and $Z^{ddl}$ into the original value function (\ref{value function t}), we obtain a new function $\widetilde{V}\left( {{Z^\text{ddl}}} \right)$ which absorbs the portion of the cost corresponding to the time beyond the deadline $d$, with $Z^{ddl}$ as the independent variable. Consequently, the new optimization objective is formulated as
\begin{equation}
\max \widetilde{V}\left( {{Z^\text{ddl}}} \right) - {C^\text{ddl}}\left( {n_{t}^o,n_{t}^s} \right).
\label{new utility}
\end{equation}

\section{Online GPU Provisioning Algorithms}
\label{section 4}

In this section, we propose two online GPU provisioning algorithms, addressing both predictive and non-predictive scenarios.

\subsection{Online Algorithm for Predictive Scenarios}

We adopt the Committed Horizon Control (CHC) algorithm~\cite{chen2016using} as our foundation and tailor it to the specific characteristics of our problem, resulting in our algorithm: Adaptive Hybrid Allocation with Prediction (AHAP). We select CHC from existing prediction-based online optimization methods for its effective balance between decision stability and adaptability under imperfect forecasts. CHC optimizes over a finite prediction horizon but commits only to the initial subset of decisions, thereby reducing the risk of excessive reliance on uncertain predictions. Unlike Receding Horizon Control (RHC)~\cite{mayne2000constrained}, which is sensitive to prediction errors, and Averaging Fixed Horizon Control (AFHC)~\cite{lin2012online}, which suffers from error accumulation, CHC provides a tunable trade-off between responsiveness and robustness. Hence, CHC is especially well-suited for our dynamic environment characterized by noisy and volatile predictions.

CHC involves two key hyperparameters: the \textit{prediction window} \( \omega \) and the \textit{commitment level} \( v \).
The prediction window \( \omega \) determines the forecasting horizon at each time slot, providing a predicted sequence of length \( \omega \). This allows the algorithm to generate a sequence of decisions by maximizing the cumulative objective over the prediction window. The commitment level \( v \) determines the portion of the decision sequence that is actually implemented. Specifically, in each time slot, only the first \( v \) decisions in the sequence are executed, meaning that the actual decision depends on the current and previous \( v{-}1 \) decision sequences.

To further adapt CHC to our setting, we introduce an additional hyperparameter \( \sigma \), representing the \textit{spot price threshold}. In our design, as long as the job has not reached its deadline, all spot instances priced below \( \sigma \) are utilized. This approach allows us to aggressively leverage low-cost spot instances to minimize overall resource costs. The introduction of $\sigma$ constitutes a key theoretical contribution, as it introduces an additional source of prediction error in our model, which directly affects the performance bound. This distinction allows our algorithm to better handle the volatility of spot markets while remaining cost-efficient.

\begin{algorithm}[t]
    \caption{Adaptive Hybrid Allocation for Predictive Scenarios (AHAP)}
    \label{algorithm:1}
    \KwIn{\{$L$,$d$,$N^{min}$,$N^{max}$\},\{$p^o$, $p_t^s$, $n_{t}^{avail}$ \},$\omega, v, \sigma$.}
    \KwOut{$u$, $n_{t}^{o}$, $n_{t}^{s}$, $t \in \left[ {1,{d}} \right]$ .}
    Initialize $Z_{0}=0,n_{0}=0$\;
    \For {$t=1,2,\dots, d$} {
        Predict spot price $p_{ \cdot \left| t \right.}^s$and availability $n_{\cdot \left| t \right.}^{avail}$\;
        Calculate expected progress $Z_{t+\omega}^{\exp}$ at current time slot according to (\ref{expected progress})\;
        
        \If{$Z_{t-1} \ge Z_{t+\omega}^{\exp}$} {
          \For{$\tau  = 0,1, \ldots \omega$ }
          {
          \If{$p_{t + \tau \left| \tau  \right.}^s \le \sigma  \cdot {p^o}$and $n_{t + \tau \left| t \right.}^{avail} \ge N^{\min }$}
             {$n_{t + \tau }^{s,t} \leftarrow \min \left\{ {n_{t + \tau \left| t \right.}^{avail},N^{\max }} \right\}$\;}
          \Else{
             $n_{t + \tau }^{s,t} \leftarrow 0$\;
          }
          $n_{t + \tau }^{o,t} \leftarrow 0$\;
          }
        }
        \Else{
           Solve the problem (\ref{partial object}) and get solutions $\{n_{\tau}^{o,t}\}$, $\{n_{\tau}^{s,t}\}$,$\tau  \in \left[ {t,t + \omega } \right]$\;
         }
           $n_{t}^o \leftarrow \sum\limits_{k = 0}^{v - 1} {n_{t}^{o,t - k}} $\;
           $n_{t}^s \leftarrow \min \{\sum\limits_{k = 0}^{v - 1} {n_{t}^{s,t - k}},n_t^{avail}\} $\;

           Limit $n_{t}^o+n_{t}^s$ in range of $\left[ {N^{\min },N^{\max }} \right]$\;
        
        Process the job and update the job progress $Z_{t}\leftarrow Z_{t-1} + \eta_t \cdot H(n_{t}^{s} + n_{t}^{o})$\; 

        \If{$Z_{t} \ge L$} {
          \textbf{break} from line 2\;
        }
    }
    Calculate utility $u$ according to (\ref{new utility})\;   
    \textbf{return} allocation decisions and job utility $u$\;
\end{algorithm}

The AHAP algorithm, presented in Algorithm \ref{algorithm:1}, operates as follows. At each time slot $t$, we compute the future $\omega$ time slots' forecasts and expected job progress (Lines 3–4). The current job progress then is compared against the expected progress. If the current progress exceeds the expected progress (Lines 5–11), job allocation within the prediction window prioritizes spot instances with prices below the threshold $\sigma$. If the current progress lags behind the expected progress (Lines 12–13), the CHC framework is applied to compensate for the shortfall within the prediction window by solving the following problem to obtain the optimal allocation sequence $\{n_{t}^{o,t}, \ldots , n_{t+\omega }^{o,t}\}$ and $\{n_{t}^{s,t}, \ldots , n_{t+\omega }^{s,t}\}$:
\begin{equation}
    \mathop {\max }\limits_{\substack{\{n_{\tau}^{o,t}\},
    \{n_{\tau}^{s,t}\} \\
    \tau  \in \left[ {t,t + \omega } \right]}}
    \widetilde{V}\left( Z_{t+\omega }^t \right)
    - \sum\limits_{\tau = t}^{t+\omega} \left( n_{\tau }^{o,t} \cdot p^o + n_{\tau }^{s,t} \cdot p_{\tau \vert t}^s \right).
    \label{partial object}
\end{equation}
We store the optimal allocation sequences at each time slot. For instance, at time slot $t-1$, the optimal allocation sequence for the $\omega+1$ time slots is recorded as $\{n_{t-1}^{o,t-1}, \ldots , n_{t+\omega-1 }^{o,t-1}\}$ and $\{n_{t-1}^{s,t-1}, \ldots , n_{t+\omega-1 }^{s,t-1}\}$.
The final allocation decision at time slot $t$ is determined by averaging the allocations over the past $v$ time slots (Lines 14–16). The job is then executed based on the allocation results (Line 17). AHAP repeats this process for each time slot $t$ until the job reaches its deadline, yielding the complete resource allocation plan and the final job utility $u$. This framework ensures that AHAP effectively adapts to dynamic cloud environments by leveraging historical commitments, predictive results, and price-threshold-driven optimizations, thereby achieving efficient and cost-effective GPU resource allocation.

We analyze the performance of the AHAP algorithm and establish an upper bound on the utility difference between AHAP and the offline optimal strategy. In particular, we consider the impact of prediction errors when forecasting. Since multi-step predictions tend to accumulate errors over time, and prediction errors generally increase as the prediction window lengthens, we formally define the upper bound on the utility function's prediction error as follows.
\begin{defi}
\label{definition:1}
The sequence of true data $ {\boldsymbol{{y_{{\omega+1}}}}, \ldots ,\boldsymbol{y_{{d}}}} $ and their $\omega$-step ahead predictions $ {\boldsymbol{y_{{\omega+1}\left| 1 \right.}}, \ldots ,\boldsymbol{y_{{d}\left| {d}-\omega \right.}}} $ in function $u\left(  \cdot  \right)$ satisfy a $\omega$-step prediction budget ${G_{\omega,d}}$ if
\begin{equation}
    \sum\limits_{t = \omega+1}^d {\mathop {\sup }\limits_{\boldsymbol x \in \Delta } \left| {u\left( {\boldsymbol x,\boldsymbol{y_t}} \right) - u\left( {\boldsymbol x,\boldsymbol{y_{t\left| {t - \omega} \right.}}} \right)} \right|}  \le {G_{\omega,d}}.
    \label{prediction budget}
\end{equation}
\end{defi}

Furthermore, we investigate the relationship between prediction error and the price threshold $\sigma$. We assume that at prediction level $\omega$, the predicted spot availability with price below threshold $\sigma$ is limited to the  range $\left[ {0,{D_{\omega,\sigma }}} \right]$.

Theorem \ref{theorem:1} establishes a performance guarantee for the AHAP algorithm under the given assumptions on the utility function's prediction error. 
Smaller prediction errors yield a tighter upper bound on the performance gap. 
The bound alse captures two key trade-offs: a higher commitment level $v$ enhances stability but reduces responsiveness, while a lower price threshold $\sigma$ improves cost efficiency at the risk of spot shortages. Notably, the scenario-specific hyperparameter $\sigma$ introduces a prediction error, widening the performance gap from optimality compared to the original CHC algorithm. This distinction forms the core contribution of our algorithm, allowing it to better address the dynamic challenges of volatile spot markets.
\begin{theorem}
\label{theorem:1}
Assuming that the prediction budget of the utility function $U$ follows (\ref{prediction budget}), then for algorithm \ref{algorithm:1}, we have:
{\small
\begin{equation}
    \sup\left\{U\left({OPT}\right)-U\left({AHAP}\right)\right\} \le \frac{2}{v}\sum\limits_{k = 1}^v {{G_{k,d}} + \frac{{\sigma {p^o}d}}{v}}\sum\limits_{k = 1}^v {{D_{k,\sigma }}} 
    \label{chc regret}
\end{equation}
}
\end{theorem}
\begin{proof}[Sketch of Proof]
    We sketch the main steps of the analysis and defer the full proof in Appendix B. The decision horizon $d$ is segmented into fixed-length windows, within which local scheduling decisions $( {\boldsymbol{x}_{\tau v + 1}, \ldots ,\boldsymbol{x}_{\left( {\tau  + 1} \right)v}} )$ are computed based on predicted workloads. For each window, we bound the utility gap caused by prediction errors $G_{v,d}$ and $D_{v,\sigma}$. To facilitate analysis, we construct a sequence of hybrid policies by gradually replacing segments of the offline optimal solution $( {\boldsymbol{x}_{\tau v + 1}^*, \ldots ,\boldsymbol{x}_{\left( {\tau  + 1} \right)v}^*} )$ with these local predicted solutions. By summing the per-window gaps and averaging across all segments, we derive an overall regret bound that tightens with decreasing prediction error under a bounded prediction budget.
\end{proof}

\subsection{Online Algorithm for Non-predictive Scenarios}

To address the limitations of AHAP under large prediction errors, as indicated by Theorem \ref{theorem:1}, we introduce the Adaptive Hybrid Allocation for Non-Predictive Scenarios (AHANP) as a fallback strategy. Rather than serving as a zero-horizon variant of AHAP, AHANP is a distinct reactive algorithm tailored for settings with poor or unavailable predictions. It bypasses explicit optimization and instead relies on interpretable per-slot metrics to guide decisions, enabling rapid adaptation to external dynamics and progress deviations.
Specifically, we introduce three indicators: workload progress \( \hat{z} = z_{t-1} / z_{t-1}^{\exp} \), spot price ratio \( \hat{p} = p_t^s / (\sigma \cdot p^o) \), and availability change rate \( \hat{n} = n_t^{\text{avail}} / n_{t-1}^{\text{avail}} \). These guide adaptive decisions that (1) ensure progress toward deadlines, (2) prefer low-cost spot instances, and (3) reduce reconfiguration by promoting allocation stability.

The AHANP algorithm is presented in Algorithm 3, with full pseudocode available in Appendix A. Specifically, for each time slot $t$, AHANP computes the expected progress and gets $\hat{z},\hat{n}$ and $\hat{p}$ based on the current spot price and availability (Line 3). Then AHANP assigns the total number of instances $n_{t}$ for time slot $t$ according to following conditions (Line 4). If progress is ahead of schedule and no spot instances are available, the job remains idle (case 1). If the availability of spot instances decreases significantly, the total number of instances is reduced (case 2). If spot availability is stable but price is high, prior allocation is reused to avoid reconfiguration. (case 3\&4). If the spot price is low, all available spot instances are utilized (case 5). When the progress is behind schedule, the instance count is doubled (case 6\&7). Then AHANP allocates as many instances as possible to spot instances, with the remaining assigned to on-demand instances (Lines 6–7). Based on this allocation, it then computes job execution (Line 8). This process is repeated for each time slot \( t \) until the deadline is reached, ultimately producing a complete resource allocation plan and the final job utility $u$.

This heuristic approach optimizes resource allocation without reliance on future predictions while effectively balancing cost efficiency and workload performance.

\section{Online Policy Selection}
\label{online learning algorithm}

In this section, we construct a policy pool by varying hyperparameters of the algorithms proposed in Section~\ref{section 4}, and develop Online Policy Selection Algorithm to identify the optimal policy from the pool.

\subsection{Policy Pool}

The online GPU provisioning algorithms in Section~\ref{section 4} require hyperparameter tuning, as they may not generalize across dynamic environments with volatile spot prices and availability. Sudden data shifts (e.g., maintenance or holidays) can also degrade prior policies.

To address this, we construct a policy pool that combines multiple GPU provisioning algorithms with systematically varied hyperparameter settings to accommodate diverse environments. Specifically, AHAP policies are parameterized by prediction window length $\omega$, commitment level $v$, and price threshold $\sigma$, while AHANP strategies adjust the price threshold $\sigma$. The pool remains extensible to support additional algorithms and configurations, enabling robust adaptation to changing workload dynamics.

\subsection{Algorithm Overview}

We present the process of our Online Policy Selection Algorithm in Algorithm \ref{algorithm:online learning}. We consider that there are $K$ various fine-tuning jobs and we leverage the candidate policies in the policy pool to process them. The policy pool includes $M$ candidate policies. The goal of the algorithm is to search for the best policy among the candidate policy pool under different scenarios. After a job is finished with a certain allocation policy, we can calculate the utility of this candidate policy $m$ for the particular job $k$, denoted as $u_k^m$. We represent the utility vector as $\boldsymbol{u_k}=[u_k^1, u_k^2, \cdots, u_k^m, \cdots, u_k^M]$. Our algorithm works by learning a weight vector $\boldsymbol{w_k}$. The weight vector is normalized $\sum_{i=1}^{M}w_k^i=1$, and can also be viewed as a distribution from another perspective. Each element $w_{k}^{m}$ of $\boldsymbol{w_k}$ indicates the weight of candidate policy $m$, and a higher weight means that this candidate policy has generally better performance in terms of the loss. We define that completing a job represents one iteration. After each iteration $k$, we update the weight vector $\boldsymbol{w_{k+1}}$ based on the utility vector $\boldsymbol{u_k}$ we collect. The update rule follows Exponentiated Gradient (EG)~\cite{hazan2016introduction, orabona2019modern}. The learning rate $\eta$ controls the update step size, and we can obtain strong theoretical guarantees by setting a proper $\eta$. The online learning algorithm aims to minimize the regret. After a certain learning iterations, the algorithm should converge and the learned weight vector converges to a sparse vector, where the candidate policy with the highest weight stands out as the best policy.

\subsection{Theoretical Analysis}

\begin{algorithm}[t]
    \caption{Online Policy Selection Algorithm}
    \label{algorithm:online learning}
    \KwIn{A set of $M$ candidate policies.}
    \KwOut{The policy of the best performance among the candidate policy pool.}
    The weights of each policy make up a vector $\boldsymbol{w_k} \in \{ \boldsymbol{w} \in \mathbb{R}^{M}: w^i > 0, \| \boldsymbol{w} \|_{1} = 1 \}$\;
    Initialize the weight vector $\boldsymbol{w_1}=[\frac{1}{M}, \frac{1}{M}, \cdots, \frac{1}{M}]$\;
    Set the learning rate $\eta = \sqrt{\frac{2 \ln(M)}{K}}$\;
    \For {$k=1, 2, \dots, K$} {
       Receive job $k$\;
       Select the resource allocation policy for job $k$ according to the current weight vector $\boldsymbol{w_k}$\;
       \For{$m=1, 2, \dots, M$} {
         Calculate the utility $u_{k}^{m}$, assuming we apply allocation policy $m$ to finish the job $k$\;
       }
       \For{$m=1, 2, \dots, M$} {
         Update the weights $w_{k+1}^{m} \leftarrow \frac{w_{k}^{m}\exp(\eta \cdot u_{k}^{m})}{\sum_{i=1}^{M} w_{k}^{i} \exp(\eta \cdot u_{k}^{i})}$\;
       }
       The new weight vector $\boldsymbol{w_{k+1}}=[w_{k+1}^1, w_{k+1}^2, \cdots, w_{k+1}^{M}]$ is learned\;
        
    }
    \textbf{return} weight vector $\boldsymbol{w_K}$\;
\end{algorithm}

We make a theoretical analysis for the regret bound of our Online Policy Selection Algorithm. The average regret of our algorithm scales with the number of candidate policies $M$ and decays with the number of iterations $K$. As $K \rightarrow \infty$, the average regret will approach 0.

\begin{theorem}
\label{theorem:2}
Suppose the utility function $u$ is normalized and we pick the learning rate $\eta = \sqrt{\frac{2 \ln(M)}{K}}$. Then for our online learning algorithm shown in Algorithm \ref{algorithm:online learning}, we have
\begin{equation}
    \max_{\boldsymbol{y}}\sum_{k=1}^K u_k(\boldsymbol{y}) - \sum_{k=1}^K \mathbb{E}_{\boldsymbol{w}_k}[u_k] \le {\sqrt{2K\ln(M)}}.
    \label{proving regret bound}
\end{equation}
\end{theorem}

\begin{proof}[Sketch of Proof]
    We outline the key steps of the regret analysis and defer full details to Appendix C. Our algorithm follows a multiplicative weights update. Tracking the KL divergence $\Phi_k = D_{\mathrm{KL}}(\boldsymbol{y} \| \boldsymbol{w}_k)$, We bound the per-round change in $\eta \left( \mathbb{E}_{\boldsymbol{w}_k}[u_k] - \langle \boldsymbol{y}, u_k \rangle \right) + \eta^2$. By summing over all $K$ iterations, and choosing the learning rate $\eta = \sqrt{\frac{2\ln M}{K}}$, we obtain the total regret bound of $\mathcal{O}(\sqrt{K \ln M})$, which guarantees sublinear average regret as $K$ grows.
\end{proof}

\begin{figure*}
    \centering
    \begin{minipage}[t]{0.24\textwidth}
        \centering
        \includegraphics[width=\linewidth]{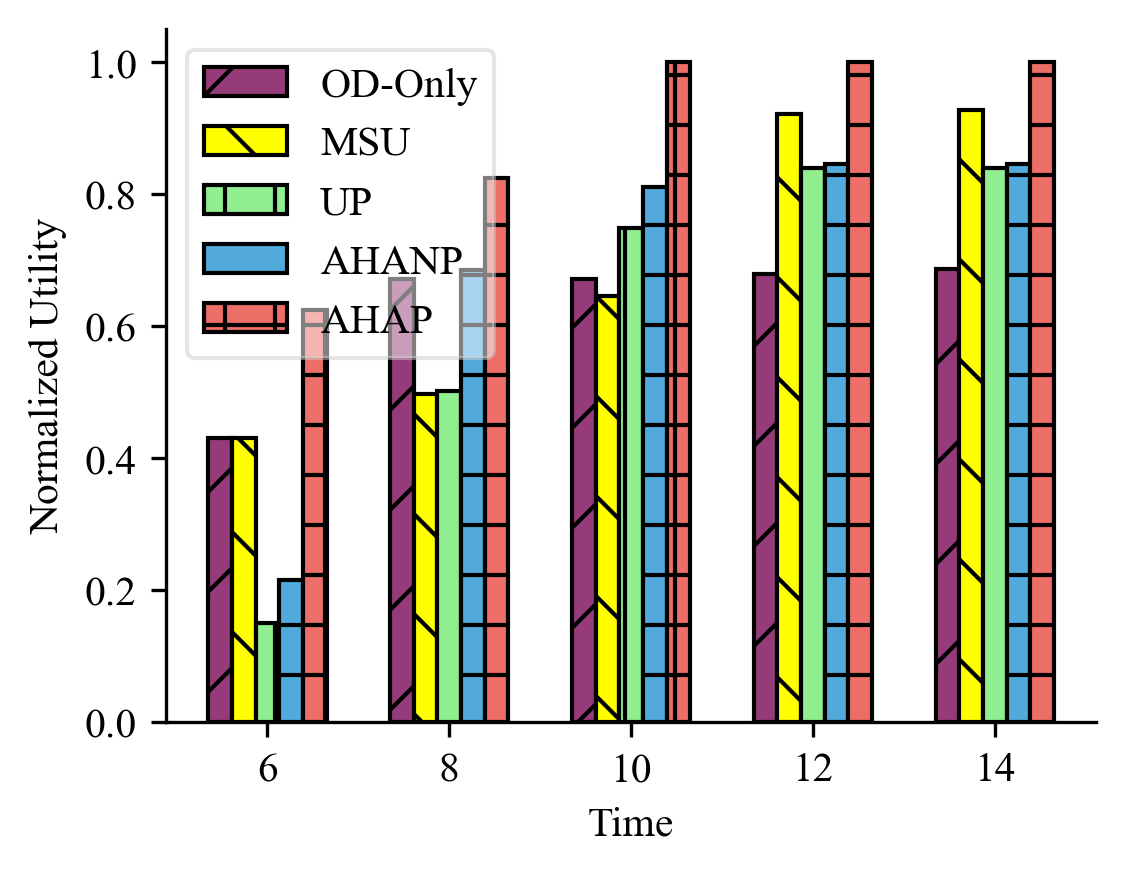}
        \captionof{figure}{Impact of Deadline.}
        \label{fig:com_ddl}
    \end{minipage}
    \hfill
    \begin{minipage}[t]{0.24\textwidth}
        \centering
        \includegraphics[width=\linewidth]{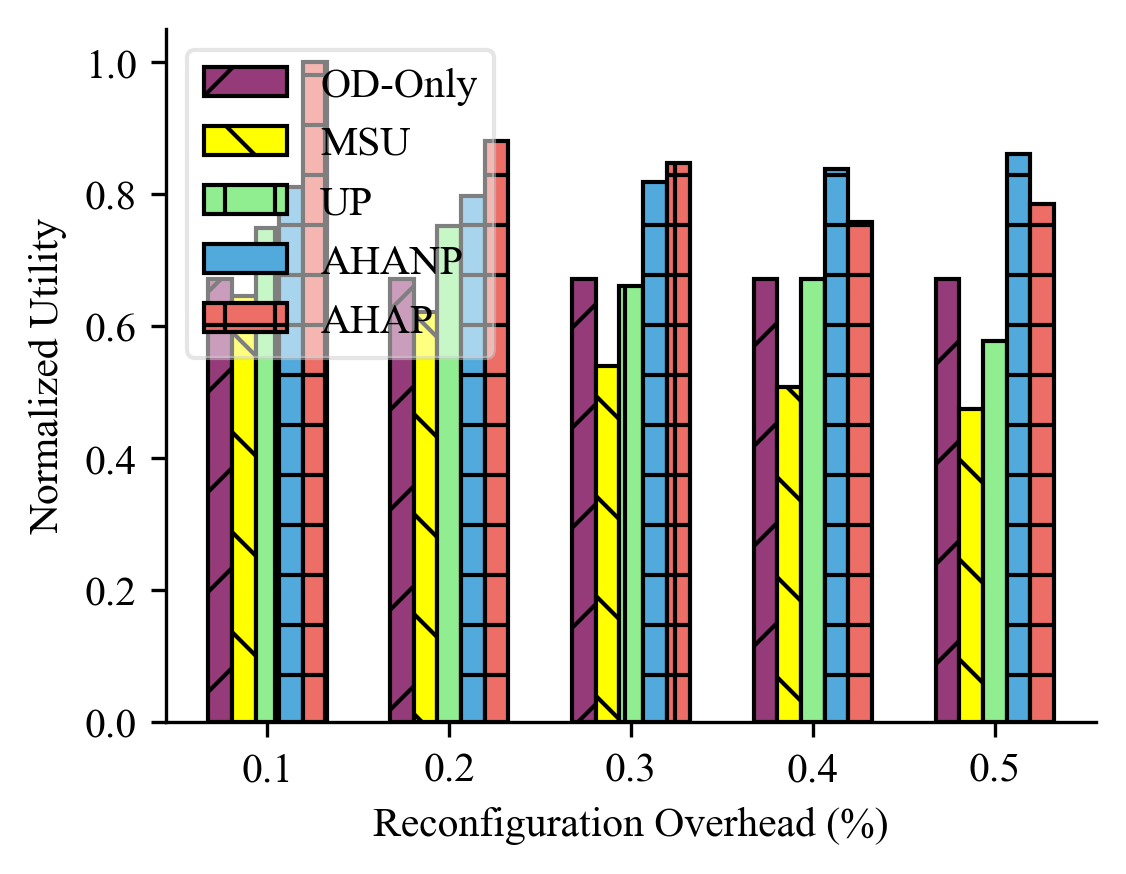}
        \captionof{figure}{Impact of Reconfiguration Overhead.}
        \label{fig:com_overhead}
    \end{minipage}
    \hfill
    \begin{minipage}[t]{0.24\textwidth}
        \centering
        \includegraphics[width=\linewidth]{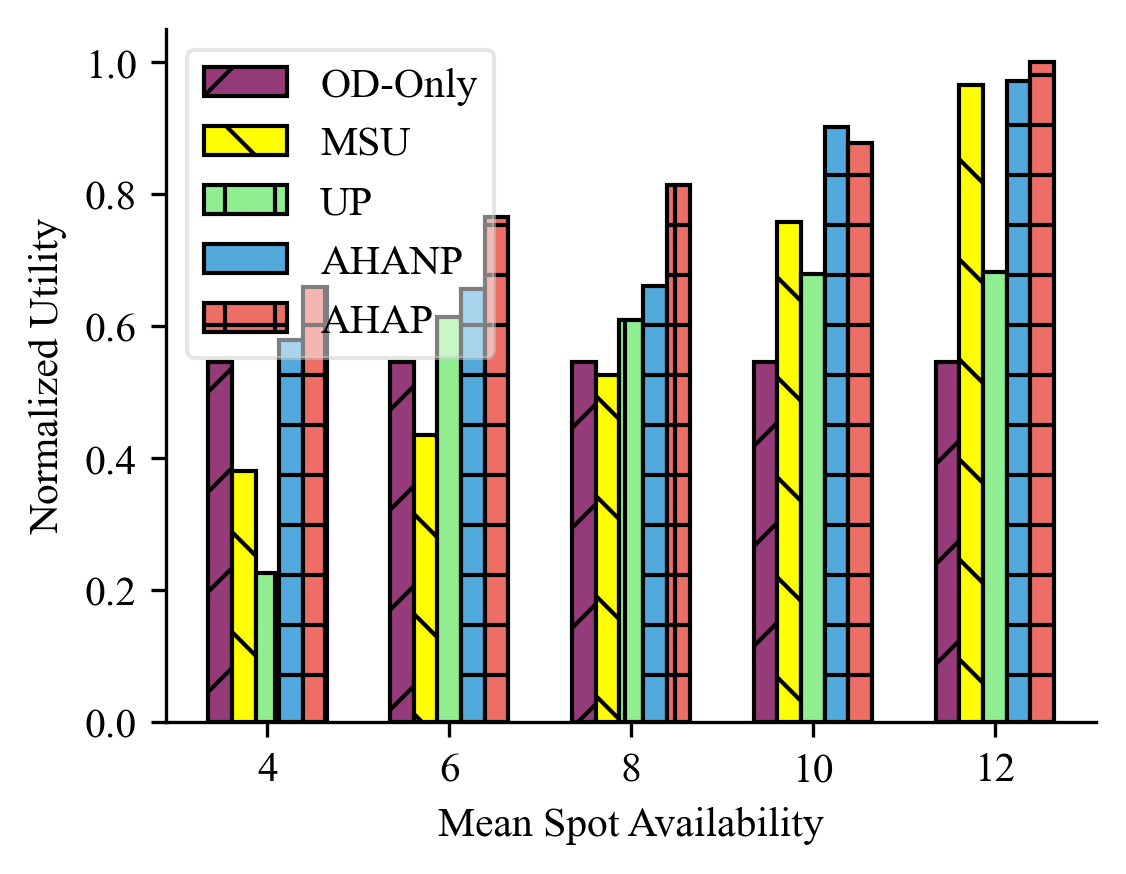}
        \captionof{figure}{Impact of Mean Spot Availability.}
        \label{fig:com_avail}
    \end{minipage}
    \hfill
    \begin{minipage}[t]{0.24\textwidth}
        \centering
        \includegraphics[width=\linewidth]{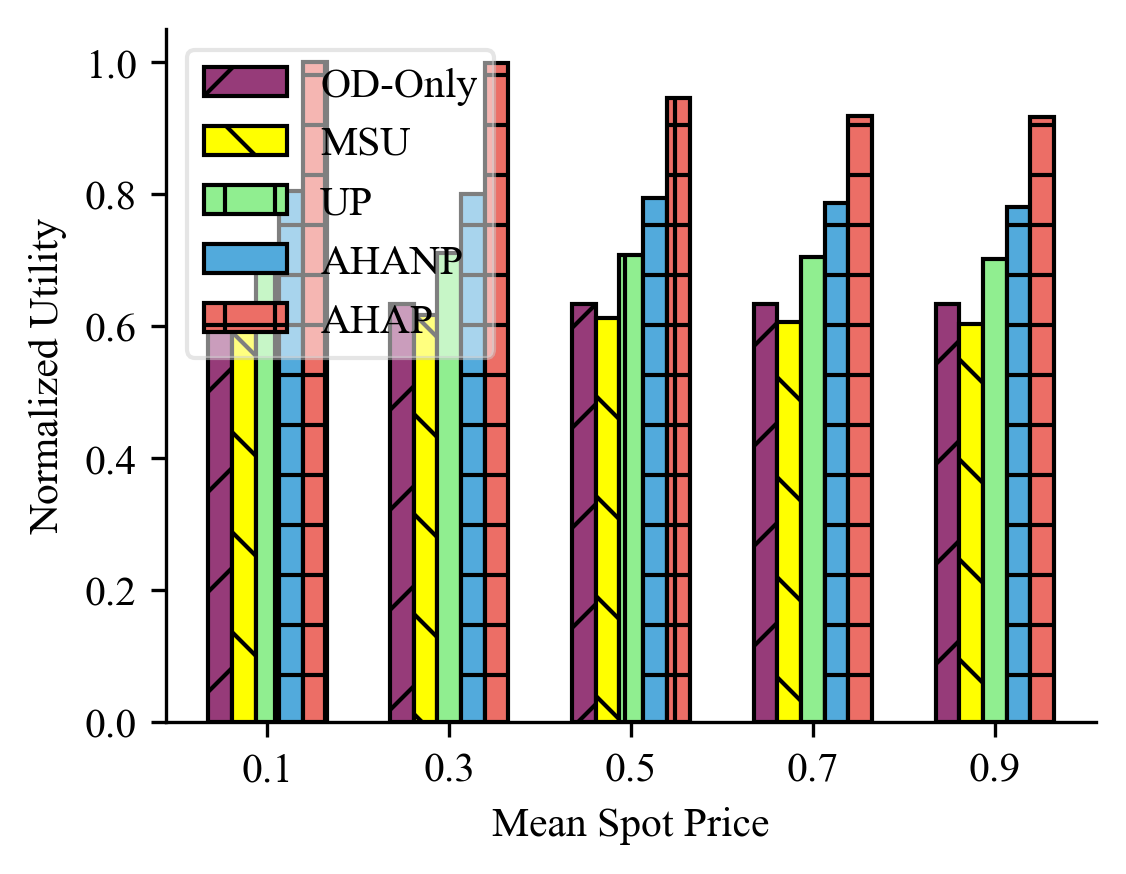}
        \captionof{figure}{Impact of Mean Spot Price.}
        \label{fig:com_price}
    \end{minipage}
    \label{fig:compare}
\end{figure*}

\section{Experimental Evaluation}

This section evaluates the proposed algorithms via simulations to validate their effectiveness, robustness, and adaptability in dynamic environments.

\begin{figure*}[htbp]
  \centering
  \begin{subfigure}[b]{0.24\textwidth}
    \centering
    \includegraphics[width=\linewidth]{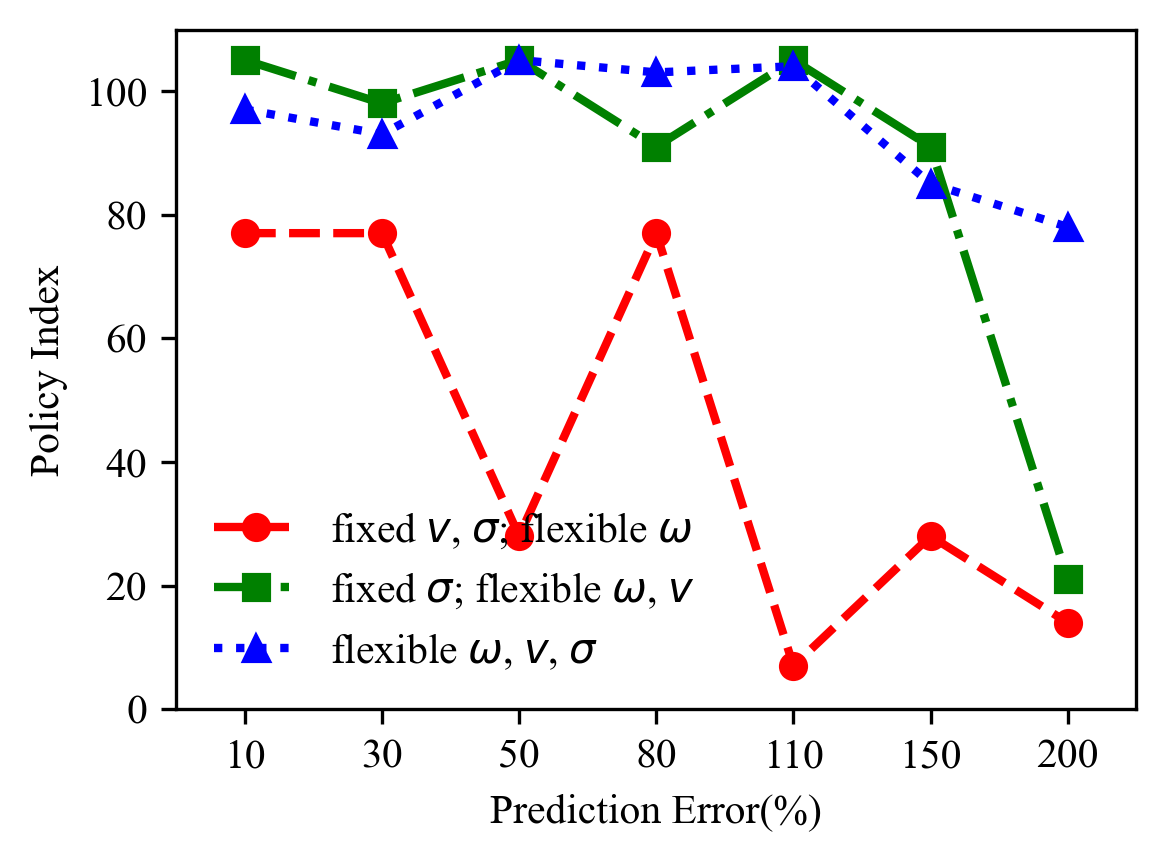}
    \caption{Mag-Dep. + Uniform}
    \label{fig5:sub1}
  \end{subfigure}
  \hfill
  \begin{subfigure}[b]{0.24\textwidth}
    \centering
    \includegraphics[width=\linewidth]{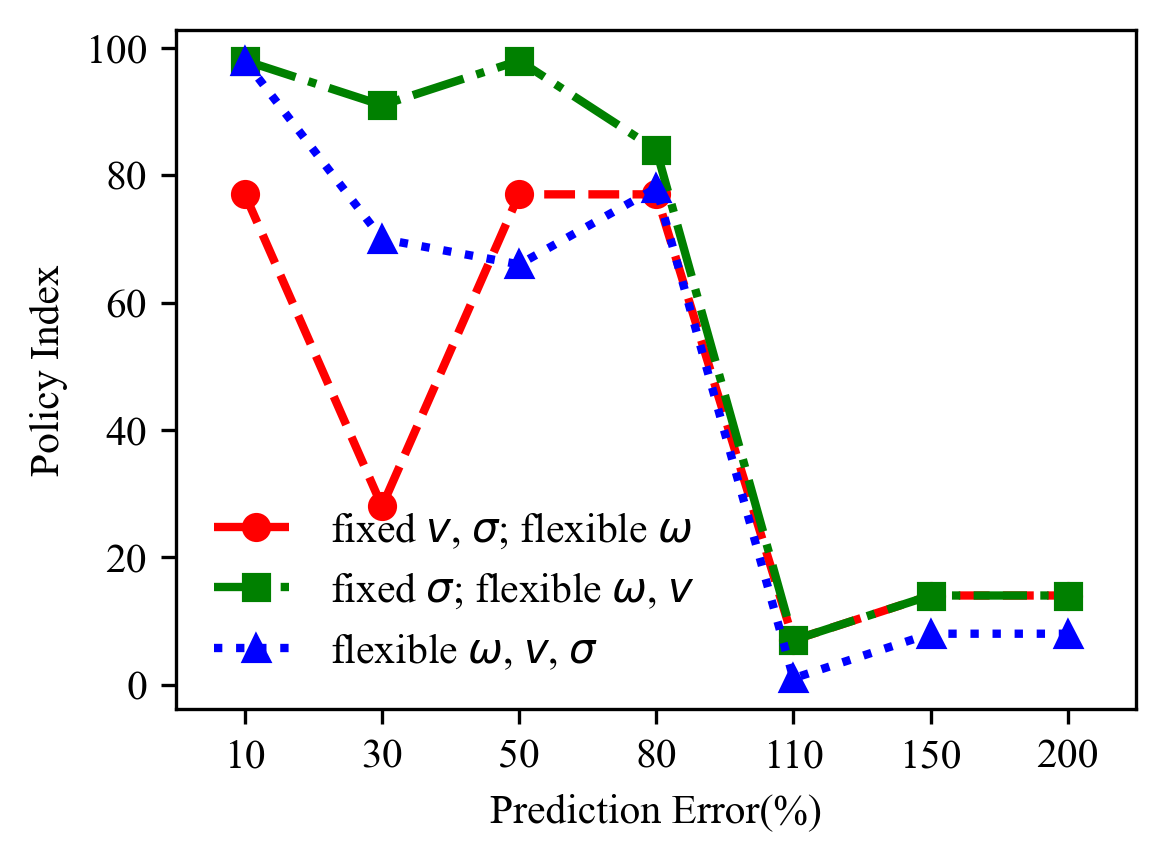}
    \caption{Fixed-Mag. + Uniform}
    \label{fig5:sub2}
  \end{subfigure}
  \hfill
  \begin{subfigure}[b]{0.24\textwidth}
    \centering
    \includegraphics[width=\linewidth]{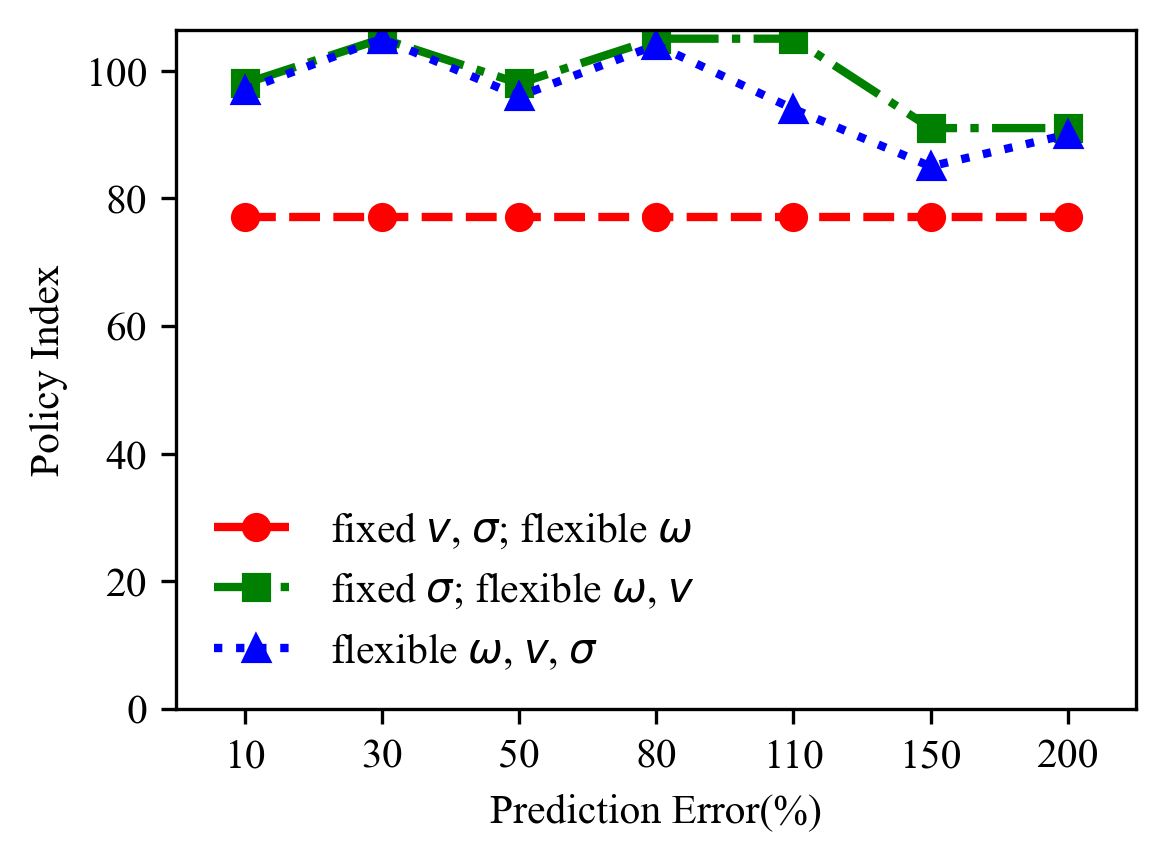}
    \caption{Mag-Dep. + Heavy-Tail}
    \label{fig5:sub3}
  \end{subfigure}
  \hfill
  \begin{subfigure}[b]{0.24\textwidth}
    \centering
    \includegraphics[width=\linewidth]{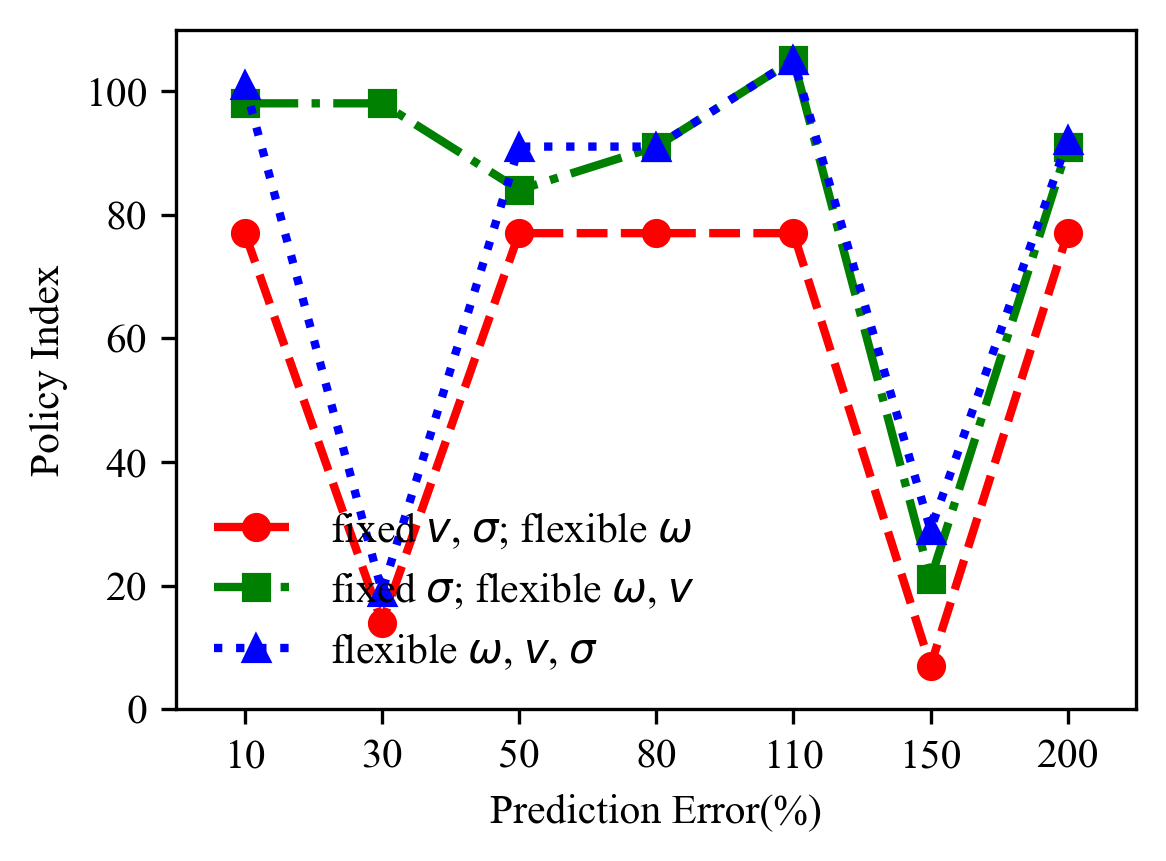}
    \caption{Fixed-Mag. + Heavy-Tail}
    \label{fig5:sub4}
  \end{subfigure}

  \caption{Convergence of Online Policy Selection under Different Prediction Noises and Hyperparameter Settings.}
  \label{fig5:all}
  \vspace{-1em} 
\end{figure*}

\subsection{Evaluation Settings}

\textbf{Cloud Service.} We simulate fine-tuning jobs on NVIDIA A100 GPUs in 30-minute time slots. Spot price traces are normalized from Vast.ai data with the on-demand price set to 1. Regional spot availability is approximated by uniformly downscaling global Vast.ai availability and capped within $\left[0, 16\right]$. To limit reconfiguration overhead, we pre-select instances with both upload and download bandwidths above 800 Mbps. While bandwidth variations impact multi-GPU scaling and reconfiguration costs, they do not alter the analytical framework.

\textbf{Fine-Tuning Jobs.} We consider fine-tuning jobs based on the LLaMA2-7B model using LoRA with a rank of 16 and a training dataset of 20 million tokens for one epoch. Such a job requires approximately 5 hours on 8 A100 GPUs~\cite{dettmers2023qlora}, corresponding to 10 time slots, , which we use as a reference deadline. Assuming unit GPU compute power equals 1, the total workload is simplified to 80. 
We use half-precision storage, and the model and data fit within the 80GB memory of a single A100 GPU; therefore, we set the minimum parallelism to \( N^{\min}=1 \). We set the maximum parallelism \( N^{\max} = 12 \) based on observations of communication overhead in multi-GPU setups \cite{nccl_user_guide}.
When reconfiguration is required, launching a new instance takes approximately 3 minutes under an 800 Mbps bandwidth. Accordingly, the proportion of effective computation time within the time slot is set to \( \mu = 0.9 \).

\textbf{Baselines.} We compare the proposed AHAP and AHANP algorithms against the following approaches: the On-Demand Only Policy (OD-Only), the Maximal Spot Utilization Policy (MSU) and Uniform Progress Policy (UP)~\cite{wu2024can}. OD-Only executes fine-tuning jobs using only on-demand instances. MSU prioritizes using all available spot instances in the early stages of job execution and switches to on-demand instances near the deadline. UP, though designed for static spot prices and synchronized availability changes, remains a relevant baseline given its adaptability to dynamic environments. At each time slot, UP compares the accumulated progress against the expected progress, which incorporates reconfiguration overhead. It prioritizes spot usage when available, and resorts to on-demand instances only when progress falls behind and spot instances are unavailable.

\textbf{Policy Pool.} To construct the policy pool, we generate 105 AHAP and 7 AHANP policies by systematically varying algorithmic hyperparameters. Specifically, for AHAP, we consider prediction window lengths $\omega \in \{1, 2, 3, 4, 5\}$. For each $\omega$, the commitment level $v$ takes integer values in $[1, \omega]$, yielding a total of 15 $(\omega, v)$ combinations. For each of these, we vary the price threshold $\sigma \in \{0.3, 0.4, \dots, 0.9\}$, resulting in $15 \times 7 = 105$ distinct AHAP policies. Here, $\omega$ controls how far to predict, $v$ determines how many predicted decisions are applied, and $\sigma$ sets the threshold below which spot instances are considered cheap and fully utilized. For AHANP, only the price threshold $\sigma$ is varied over the same 7 values, producing 7 policies.

\textbf{Prediction Noise.} To compare online algorithm convergence under different prediction errors, we define four prediction noise settings. Noise is either magnitude-dependent (Mag-Dep.) or fixed-magnitude (Fixed-Mag.), and follows either a uniform or heavy-tailed distribution. This yields four scenarios: Mag-Dep. + Uniform, Fixed-Mag. + Uniform, Mag-Dep. + Heavy-Tail, and Fixed-Mag. + Heavy-Tail.

\subsection{Evaluation Results}

\textbf{Impact of Jobs and Cloud Resource Dynamics.} Since the Online Policy Selection Algorithm can identify the better-performing policy between AHAP and AHANP, we compare each of them individually against the baselines using normalized utility, where the selected optimal policy is always the better of the two. 
Figure~\ref{fig:com_ddl} illustrates how varying deadlines affect the utility of fine-tuning jobs. AHAP consistently outperforms all baselines under both tight and relaxed deadlines. At a representative setting with deadline $=10$, AHAP improves utility by 49.0\%, 54.8\%, 33.4\%, and 23.2\% over OD-Only, MSU, UP, and AHANP, respectively.
Figure~\ref{fig:com_overhead} shows the impact of reconfiguration overhead on utility, simulated by varying network bandwidth from 100~Mbps to 800~Mbps. AHAP and AHANP consistently outperform other baselines. As reconfiguration overhead increases, all algorithms suffer utility degradation, except AHANP, which maintains stable performance. This robustness stems from its design principle of maintaining a similar number of active instances across time slots, thereby minimizing reconfiguration overhead.
Figures~\ref{fig:com_avail} and~\ref{fig:com_price} analyze the impact of average spot instance availability and price fluctuation, respectively. AHAP and AHANP remain among the top-performing algorithms across all settings. Overall, AHAP delivers the most stable performance, while AHANP exhibits advantages under limited bandwidth by avoiding frequent reconfigurations.

\textbf{Convergence under Prediction Noise.} Figure~\ref{fig5:all} illustrates the convergence of the online policy selection algorithm to the optimal policy under different types and levels of prediction noise. We also investigate the impact of fixing individual hyperparameters in the policy pool, which consists of 105 AHAP policies (e.g., fixing $v=1$ or $\sigma=0.9$). 
We simulate 1000 fine-tuning jobs with workloads uniformly distributed in $[70, 120]$, deadlines fixed at 10 time slots, and parallelism parameters combining \( N^{\min} \in [1,4] \) and \( N^{\max} \in [12,16] \).
The results show that both the type and magnitude of noise influence the optimal policy, highlighting the necessity of online adaptation in dynamic environments. Furthermore, constraining hyperparameter flexibility affects the convergence outcome, indicating that each hyperparameter in AHAP contributes to policy effectiveness. Increasing the number of flexible hyperparameters enables finer-grained policies, thereby raising the upper bound of achievable utility.

\begin{figure}[t] 
\centering
\includegraphics[width=\columnwidth]{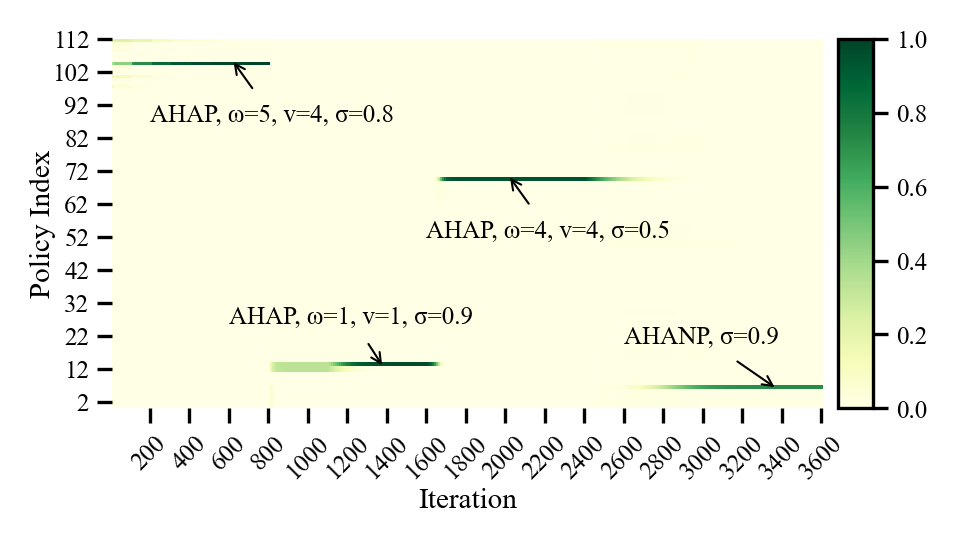} 
\caption{Policy Weight Dynamics in Online Policy Selection with Changing Prediction Quality.}
\label{fig:adaptive change}
\end{figure}

\textbf{Policy Evolution and Adaptation.} Figure~\ref{fig:adaptive change} presents a heatmap illustrating the dynamic evolution of policy weights under changing prediction environments. The prediction settings change across four phases: iterations 0–800 use Fixed-Mag. + Uniform with 10\% error; 800–1600 switch to Fixed-Mag. + Heavy-Tail with 30\% error; 1600–2400 revert to Fixed-Mag. + Uniform with 50\% error; and 2400–3600 increase the noise to 200\%. 
The policy pool includes 105 AHAP and 7 AHANP policies, indexed from 1 to 112. 
The algorithm consistently adapts to different settings and promptly converges to new optimal policies, ensuring robust and efficient resource allocation under dynamic conditions.

\section{Related work}

\textit{Spot Instance Scheduling under Preemptions.}
Spot instances offer significant cost advantages for cloud computing but pose challenges due to their inherent unreliability and risk of preemption. A key focus of recent research lies in balancing cost-efficiency with reliability under such uncertainty~
\cite{monteil2023reservation,duan2024parcae,yang2023snape,gunasekaran2022cocktail}.Several works propose deadline-aware scheduling to mitigate the impact of preemptions~\cite{wu2024can,menache2014demand,wu2023monte}. To improve robustness, redundancy and over-provisioning techniques are frequently employed~\cite{mao2025skyserve,thorpe2023bamboo}. However, these approaches do not integrate predictive information on spot price and availability for forecasting future resources.

\textit{Online Resource Scheduling in Cloud Platforms.}
Online resource scheduling refers to the process of assigning computational resources to tasks or virtual machines in real-time, without prior knowledge of future workloads. Initial work in cloud scheduling employed simple heuristic methods such as First-Come-First-Serve (FCFS) and Round Robin~\cite{buyya2009cloud,selvarani2010improved}. Recent studies have focused on reinforcement learning (RL) and other machine learning models to learn dynamic scheduling policies~\cite{liu2022deep,zhou2021multi,zhang2023lifting,gao2022titan,bao2018online}. These models adapt to workload patterns and improve over time, with promising results in both simulation and real systems. However, these methods often assume stable resource availability, limiting their effectiveness in volatile spot markets.

\textit{Prediction-Augmented Scheduling for Deadline Guarantees.}
A foundational aspect of prediction-augmented scheduling lies in accurately estimating task execution time~\cite{zhang2025deadline}. Recent approaches apply regression models, decision trees, and neural networks to learn non-linear relationships between input features and task duration~\cite{nabi2021dralba,xu2022esdnn}. To further enhance deadline adherence, some works focus on forecasting system load or future task arrivals~\cite{karim2021bhyprec,mao2019learning}. However, many existing solutions rely heavily on prediction accuracy without robust fallback mechanisms. In contrast, our approach combines forecast-driven planning with fallback resource allocation to meet deadlines cost-effectively in volatile spot markets. 

\section{Conclusions}

To minimize fine-tuning costs under deadline constraints in both predictable and unpredictable spot market conditions, we propose an analytical framework. It features a prediction-based allocation algorithm, a complementary non-predictive algorithm and further an online policy selection algorithm that adaptively chooses the best-performing policy from both. We show that the prediction-based algorithm improves with prediction accuracy, while the policy selector achieves a regret bound of $\mathcal{O}(\sqrt{T})$. Experiments demonstrate that our framework adapts effectively to varying conditions and consistently outperforms baselines, improving utility by up to 54.8\%. The framework is flexible and remains adaptable to future prediction and allocation strategies.

\clearpage
\bibliography{refs}

\clearpage
\appendix
\section{Appendix}

\subsection{Pseudocode of AHANP Algorithm}
\label{appendix:ahanp}
\begin{algorithm}
    \caption{Adaptive Hybrid Allocation Algorithm for Non-Predictive Scenarios (AHANP)}
    \label{algorithm:2}
    \KwIn{\{$L$,$d$,$N^{min}$,$N^{max}$\},\{$p^o$, $p_t^s$, $n_{t}^{avail}$ \},$\sigma$.}
    \KwOut{$u$, $n_{t}^{o}$, $n_{t}^{s}$, $t \in \left[ {1,{d}} \right]$ .}
    Initialize $Z_{0}=0,n_{0}=0$\;
    \For {$t=1,2,\dots, d$} {
        Calculate expected progress $Z_{t-1}^{\exp}$ at current time slot according to (\ref{expected progress}) and get $\hat{z},\hat{n},\hat{p}$\;

        \tcp{\emph{Assign $n_{t}$ according to $\hat{z},\hat{n}$and $\hat{p}$}}
        $n_{t} \leftarrow$
        \[
        \left\{
        \begin{array}{ll}
        0 & \text{if }\hat{z} \ge 1 , \hat{n} = 0\\
        \max \left\{0.5 \times n_{t-1}, N^{min} \right\} & \text{if } \hat{z} \ge 1 ,0 < \hat{n} \le 0.5\\
        n_{t-1} & \text{if } \hat{z} \ge 1 ,0.5 < \hat{n} \le 1\\
        n_{t-1} & \text{if } \hat{z} \ge 1,\hat{n} > 1, \hat{p} > 1 \\
        \max \left\{ n_{t-1}, n_t^{avail} \right\} & \text{if }\hat{z} \ge 1, \hat{n} > 1, \hat{p} \le 1\\
        N^{min}, & \text{if } \hat{z} < 1 ,\hat{n} = \infty\\
        2 \times n_{t-1}, & \text{if }\hat{z} < 1 ,\hat{n} < \infty
        \end{array}
        \right.
        \]
         
         Limit $n_{t}$ in range of $\left[ {N^{\min },N^{\max }} \right]$\;
        $n_{t}^s \leftarrow \min \{n_t^{avail}, n_{t}\}$\;
        $n_{t}^o \leftarrow n_{t}-n_{t}^s$\;
        Process the job and update the job progress $Z_{t}\leftarrow Z_{t-1} + \eta_t \cdot H(n_{t}^{s} + n_{t}^{o})$\; 

        \If{$Z_{t} \ge L$} {
          \textbf{break} from line 2\;
        }
    }
    Calculate utility $u$ according to (\ref{new utility})\;   
    \textbf{return} allocation decisions and job utility $u$\;
\end{algorithm}

\subsection{Proof of Theorem 1}
\label{proof 1}

\begin{proof}
We first analyze the utility difference between the algorithm \ref{algorithm:1} and the optimal result before averaging the decision results. 
Let ${N_k}$ denote the number of time steps $t = 1,2, \ldots d$ that satisfy the condition $t\bmod v = k$, and let these time points be denoted by $\left( {{t^{k,1}}, \ldots ,{t^{k,{N_k}}}} \right)$. 
Let $\boldsymbol{x}_t$ be the decision vector at time slot $t$. Let $\boldsymbol{y}_t$ be the vector to be predicted and $\boldsymbol{{y}}_{{t_1}, \ldots, {t_2}\left| t \right.}$ be the prediction results from $t_1$ to $t_2$ at time $t$. 
Let $\left\{ {\boldsymbol{x}^*} \right\} = \boldsymbol{x}_1^*, \ldots ,\boldsymbol{x}_d^*$ denote the set of offline optimal solutions.
According to algorithm \ref{algorithm:1}, at time slot $\tau v$, if current workload is less than the expected workload, we solve the following problem: 
\begin{equation}
    \max {U_{\tau v + 1, \ldots ,\left( {\tau  + 1} \right)v}}\left( {\boldsymbol{x},{\boldsymbol{y}_{\left. {\tau v + 1, \ldots ,\left( {\tau  + 1} \right)v} \right|\tau v}}} \right),
    \nonumber
\end{equation}
and get the locally optimal solutions $\left( {\boldsymbol{x}_{\tau v + 1}^k, \ldots ,\boldsymbol{x}_{\left( {\tau  + 1} \right)v}^k} \right)$.
The locally optimal solutions, prediction results, and the locally utility function are denoted by $\boldsymbol{x}_{\tau ,v}^k$, ${\boldsymbol{y}_{\left.  \cdot  \right|\tau v}}$, and ${U_{\tau ,v}}\left( {\boldsymbol{x},{\boldsymbol{y}_{\left.  \cdot  \right|\tau v}}} \right)$, respectively.

We define ${N_k}$ sequences $\left\{ {\boldsymbol{\xi} ^1, \ldots ,\boldsymbol{\xi} ^{{N_k}}} \right\}$, each of length $d$, where $\boldsymbol{\xi} ^1$ is the offline optimal solutions $\left\{ {\boldsymbol{x}^*} \right\}$, and $\boldsymbol{\xi} _t^1 = \boldsymbol{x}_t^*$, for $t \in \left[ {1,d} \right]$. 
For $\tau  \in \left[ {1,{N_k}} \right]$, sequence $\boldsymbol{\xi} ^\tau $ satisfies the following update rule: 
\begin{equation}
    {\boldsymbol{\xi} ^{\tau  + 1}} = \left( {\boldsymbol{\xi} _1^\tau , \ldots ,\boldsymbol{\xi} _{\tau v}^\tau ,\boldsymbol{x}_{\tau v + 1}^k, \ldots ,\boldsymbol{x}_{\left( {\tau  + 1} \right)v}^k,\boldsymbol{\xi} _{\left( {\tau  + 1} \right)v + 1}^\tau , \ldots ,\boldsymbol{\xi} _d^\tau}  \right).
    \nonumber
\end{equation}

By the definition of local optimality under the predicted window, we have 
\begin{equation}
    U_{\tau ,v}\left( \boldsymbol{x}_{\tau ,v}^*,\, \boldsymbol{y}_{\,\cdot\,| \tau v} \right) 
    \le 
    U_{\tau ,v}\left( \boldsymbol{x}_{\tau ,v}^k,\, \boldsymbol{y}_{\,\cdot\,| \tau v} \right).
    \nonumber
\end{equation}

Hence the difference of the actual utility satisfies: 
\begin{equation}
\begin{aligned}
    & U_{\tau, v}\left(\boldsymbol{x}_{\tau, v}^*, \boldsymbol{y}_{\tau, v}\right) 
      - U_{\tau, v}\left(\boldsymbol{x}_{\tau, v}^k, \boldsymbol{y}_{\tau, v}\right) \\
\le\; & U_{\tau, v}\left(\boldsymbol{x}_{\tau, v}^*, \boldsymbol{y}_{\tau, v}\right) 
      - U_{\tau, v}\left(\boldsymbol{x}_{\tau, v}^*, \boldsymbol{y}_{\,\cdot\,| \tau v}\right) \\
 & + U_{\tau, v}\left(\boldsymbol{x}_{\tau, v}^k, \boldsymbol{y}_{\,\cdot\,| \tau v}\right) 
      - U_{\tau, v}\left(\boldsymbol{x}_{\tau, v}^k, \boldsymbol{y}_{\tau, v}\right)  + \sigma p^o \sum_{t=1}^{v} D_{t, \sigma},
\end{aligned}
\nonumber
\end{equation}
where the last term accounts for the prediction error when considering actual progress exceeding the expected progress.

Summing these inequalities from $\tau  = 1$ to $\tau  = {N_k}$, we have 
\begin{align}
& U_d\left( \boldsymbol{\xi}^1, \boldsymbol{y} \right) 
  - U_d\left( \boldsymbol{\xi}^{N_k}, \boldsymbol{y} \right) \notag 
= U(\mathrm{OPT}) 
  - U_d\left( \boldsymbol{\xi}^{N_k}, \boldsymbol{y} \right) \notag \\
\le\; & \sum_{\tau = 1}^{N_k} \left( 
  U_{\tau, v} \left( \boldsymbol{x}_{\tau, v}^*, \boldsymbol{y}_{\tau, v} \right) 
  - U_{\tau, v} \left( \boldsymbol{x}_{\tau, v}^*, \boldsymbol{y}_{\,\cdot\,| \tau v} \right) 
  \right) \notag \\
& + \sum_{\tau = 1}^{N_k} \left( 
  U_{\tau, v} \left( \boldsymbol{x}_{\tau, v}^k, \boldsymbol{y}_{\,\cdot\,| \tau v} \right) 
  - U_{\tau, v} \left( \boldsymbol{x}_{\tau, v}^k, \boldsymbol{y}_{\tau, v} \right) 
  \right) \notag \\
& + \sigma p^o \sum_{\tau = 1}^{N_k} \sum_{t = 1}^v D_{t, \sigma}. \notag
\end{align}

By summing the actual utility differences across all $k = 1, \ldots ,v$, taking the average, and applying the Jensen's inequality, we have 
\begin{align}
& U(\mathrm{OPT}) - U(\mathrm{AHAP}) \notag \\
\le\; & \frac{1}{v} \sum_{k=1}^{v} \sum_{\tau=1}^{N_k} 
\left( 
  U_{\tau, v} \left( \boldsymbol{x}_{\tau, v}^*, \boldsymbol{y}_{\tau, v} \right)
  - U_{\tau, v} \left( \boldsymbol{x}_{\tau, v}^*, \boldsymbol{y}_{\,\cdot\,| \tau v} \right)
\right) \notag \\
& + \frac{1}{v} \sum_{k=1}^{v} \sum_{\tau=1}^{N_k} 
\left( 
  U_{\tau, v} \left( \boldsymbol{x}_{\tau, v}^k, \boldsymbol{y}_{\,\cdot\,| \tau v} \right)
  - U_{\tau, v} \left( \boldsymbol{x}_{\tau, v}^k, \boldsymbol{y}_{\tau, v} \right)
\right) \notag \\
& + \frac{\sigma p^o}{v} \sum_{k=1}^{v} \sum_{\tau=1}^{N_k} \sum_{t=1}^{v} D_{t, \sigma} \notag \\
\le\; & \frac{1}{v} \sum_{k=1}^{v} \sum_{\tau=1}^{N_k} 
\left| 
  U_{\tau, v} \left( \boldsymbol{x}_{\tau, v}^*, \boldsymbol{y}_{\tau, v} \right)
  - U_{\tau, v} \left( \boldsymbol{x}_{\tau, v}^*, \boldsymbol{y}_{\,\cdot\,| \tau v} \right)
\right| \notag \\
& + \frac{1}{v} \sum_{k=1}^{v} \sum_{\tau=1}^{N_k} 
\left| 
  U_{\tau, v} \left( \boldsymbol{x}_{\tau, v}^k, \boldsymbol{y}_{\,\cdot\,| \tau v} \right)
  - U_{\tau, v} \left( \boldsymbol{x}_{\tau, v}^k, \boldsymbol{y}_{\tau, v} \right)
\right| \notag \\
& + \frac{\sigma p^o d}{v} \sum_{t=1}^{v} D_{t, \sigma} \notag \\
\le\; & \frac{2}{v} \sum_{k=1}^{v} G_{k,d}
+ \frac{\sigma p^o d}{v} \sum_{k=1}^{v} D_{k, \sigma}. \notag
\end{align}
\end{proof}

\subsection{Proof of Theorem 2}
\label{proof 2}

The online learning algorithm that we apply is an extension of OMD algorithm. The update of parameters in traditional OMD algorithm is defined as
\begin{equation}
    \boldsymbol{w_{k+1}} = \arg\min_{\boldsymbol{w} \in V} (\langle \boldsymbol{g_k}, \boldsymbol{w} \rangle + \frac{1}{\eta}B_{\psi}(\boldsymbol{w};\boldsymbol{w_k})),
    \notag
\end{equation}
where $\boldsymbol{g_k} \in \partial u_k(\boldsymbol{w_k})$, $V$ is the search space, and $B_{\psi}$ is the Bregman divergence.
In the design of our online learning algorithm, we set $\psi(\boldsymbol{w}) = \sum_{i=1}^M w_i\ln(w_i)$. Then, we can get the update rule of our online learning algorithm, as is shown in Algorithm \ref{algorithm:online learning}
\begin{equation}
    w_{k+1}^i = \frac{x_{k}^{i}\exp(\eta \cdot u_{k}^{i})}{\sum_{t=1}^{M} x_{k}^{t} \exp(\eta \cdot u_{k}^{t})}.
    \notag
\end{equation}

First, we introduce the concept of Bregman divergence and dual norm, which is prerequisite for proving Theorem \ref{theorem:2}.

\begin{defi}
\label{definition:1}
Let $\psi : X \rightarrow \mathbb{R}$ be strictly convex and differentiable on X. The \textbf{Bregman Divergence} with respect to $\psi$ is denoted by $B_{\psi}$, defined as
\begin{equation}
    B_{\psi}(\boldsymbol{x};\boldsymbol{y}) = \psi(\boldsymbol{x}) - \psi(\boldsymbol{y}) - \langle \nabla \psi(\boldsymbol{y}), \boldsymbol{x} - \boldsymbol{y} \rangle.
    \notag
\end{equation}
Bregman divergence is a similarity measure between $\boldsymbol{x}$ and $\boldsymbol{y}$. Note that the Bregman divergence is not symmetric. 
\end{defi}

\begin{lemma}
\label{lemma:1}
Let $V = \{ \boldsymbol{x} \in \mathbb{R}^M :  x_i > 0, \| \boldsymbol{x} \|_1 = 1\}$, $X = \mathbb{R}^{M}_{\ge 0}$, and $\psi(\boldsymbol{x})=\sum_{i=1}^M x_i\ln(x_i)$, the negative entropy. Then, according to the definition of the Bregman divergence in Definition \ref{definition:1}, for all $\boldsymbol{x}, \boldsymbol{y} \in X$, we have
\begin{equation}
    \begin{aligned}
    B_{\psi}(\boldsymbol{x};\boldsymbol{y}) &= \psi(\boldsymbol{x}) - \psi(\boldsymbol{y}) - \nabla \psi(\boldsymbol{y})^{\top}(\boldsymbol{x} - \boldsymbol{y}) \\
    &= \sum_{i=1}^M (x_i\ln(x_i)-y_i\ln(y_i)-(\ln(y_i)+1)(x_i-y_i)) \\
    &= \sum_{i=1}^M (x_i\ln(\frac{x_i}{y_i})-x_i+y_i).
    \end{aligned}
    \nonumber
\end{equation}
This is a special case of Bregman divergence, called as the generalized Kullback-Leibler divergence (KL divergence).
\end{lemma}

\begin{lemma}
\label{lemma:2}
Let $B_{\psi}$ the Bregman divergence w.r.t. $\psi: X \rightarrow \mathbb{R}$. Then, for any three points $\boldsymbol{x}, \boldsymbol{y}$ and $\boldsymbol{z} \in X$, the following identity holds
\begin{equation}
    B_{\psi}(\boldsymbol{z};\boldsymbol{x}) + B_{\psi}(\boldsymbol{x};\boldsymbol{y}) - B_{\psi}(\boldsymbol{z};\boldsymbol{y}) = \langle \nabla\psi(\boldsymbol{y}) - \nabla\psi(\boldsymbol{x}), \boldsymbol{z} - \boldsymbol{x} \rangle.
    \nonumber
\end{equation}
\end{lemma}

\begin{defi}
The \textbf{dual norm} $\| \cdot \|_{\star}$ of a norm $\| \cdot \|$ is defined as $\| \boldsymbol{\theta} \|_{\star}= \max_{\boldsymbol{x}:\| \boldsymbol{x} \| \le 1} \langle \boldsymbol{\theta}, \boldsymbol{x} \rangle$.
\end{defi}

Next, we present the one-step relationship.

\begin{lemma}
\label{lemma:3}
Let $B_{\psi}$ the Bregman divergence w.r.t. $\psi: X \rightarrow \mathbb{R}$, and assume $\psi$ to be proper, closed, and $\lambda$-strongly convex. Let $\boldsymbol{g_k} \in \partial u_k(\boldsymbol{w_k})$. $\forall \boldsymbol{y} \in V$, the following inequality holds
\begin{equation}
    \eta(u_k(\boldsymbol{y})-u_k(\boldsymbol{w_k})) \le B_{\psi}(\boldsymbol{w_k};\boldsymbol{y}) - B_{\psi}(\boldsymbol{w_{k+1}};\boldsymbol{y}) + \frac{\eta^2}{2\lambda} \| \boldsymbol{g_k} \|^2_{\star}.
    \nonumber
\end{equation}
\end{lemma}

\begin{proof}
First of all, from the strong convexity of $u_k$, we have
\begin{equation}
    \eta(u_k(\boldsymbol{y})-u_k(\boldsymbol{w_{k}})) \le \eta \langle \boldsymbol{g_k}, \boldsymbol{y}-\boldsymbol{w_k} \rangle.
    \nonumber
\end{equation}
We further expand the equality as
\begin{equation}
    \begin{aligned}
    & \eta \langle \boldsymbol{g_k}, \boldsymbol{y}-\boldsymbol{w_k} \rangle = \langle \nabla\psi(\boldsymbol{w_k}) -\nabla\psi(\boldsymbol{w_{k+1}}) -\eta\boldsymbol{g_k}, \boldsymbol{w_{k+1}} - \boldsymbol{y} \rangle \\
    & + \langle \nabla\psi(\boldsymbol{w_{k+1}}) -\nabla\psi(\boldsymbol{w_{k}}), \boldsymbol{w_{k+1}} - \boldsymbol{y} \rangle + \langle \eta\boldsymbol{g_k}, \boldsymbol{w_k} - \boldsymbol{w_{k+1}} \rangle \\
    &\le \langle \nabla\psi(\boldsymbol{w_{k+1}}) -\nabla\psi(\boldsymbol{w_{k}}), \boldsymbol{w_{k+1}} - \boldsymbol{y} \rangle + \langle \eta\boldsymbol{g_k}, \boldsymbol{w_k} - \boldsymbol{w_{k+1}} \rangle.
    \end{aligned}
    \nonumber
\end{equation}

    \begin{equation}
    \begin{aligned}
    \eta \langle \boldsymbol{g_k}, \boldsymbol{y}-\boldsymbol{w_k} \rangle &= \langle \nabla\psi(\boldsymbol{w_k}) -\nabla\psi(\boldsymbol{w_{k+1}}) -\eta\boldsymbol{g_k}, \boldsymbol{w_{k+1}} - \boldsymbol{y} \rangle \\
    & \quad + \langle \nabla\psi(\boldsymbol{w_{k+1}}) -\nabla\psi(\boldsymbol{w_{k}}), \boldsymbol{w_{k+1}} - \boldsymbol{y} \rangle \\
    & \quad + \langle \eta\boldsymbol{g_k}, \boldsymbol{w_k} - \boldsymbol{w_{k+1}} \rangle \\
    &\le \langle \nabla\psi(\boldsymbol{w_{k+1}}) -\nabla\psi(\boldsymbol{w_{k}}), \boldsymbol{w_{k+1}} - \boldsymbol{y} \rangle \\
    & \quad + \langle \eta\boldsymbol{g_k}, \boldsymbol{w_k} - \boldsymbol{w_{k+1}} \rangle.
    \end{aligned}
    \nonumber
\end{equation}

Using Lemma \ref{lemma:2}, we can interpret the first term as
\begin{equation}
    \begin{aligned}
    &\langle \nabla\psi(\boldsymbol{w_{k+1}}) -\nabla\psi(\boldsymbol{w_{k}}), \boldsymbol{w_{k+1}} - \boldsymbol{y} \rangle \\
    = \quad & B_{\psi}
    (\boldsymbol{w_k};\boldsymbol{y}) - B_{\psi}(\boldsymbol{w_{k+1}};\boldsymbol{y}) - B_{\psi}(\boldsymbol{w_{k+1}};\boldsymbol{w_k}).
    \end{aligned}
    \nonumber
\end{equation}
As $\psi(\boldsymbol{w})$ is $\lambda$-strongly convex, we have
\begin{equation}
\label{proof:1}
    B_{\psi}(\boldsymbol{w_{k+1}}; \boldsymbol{w_k}) \ge \frac{\lambda}{2} \| \boldsymbol{w_{k+1}} - \boldsymbol{w_k} \|^2.
    \notag
\end{equation}
According to the definition of the dual norm, we have
\begin{equation}
\label{proof:2}
    \langle \eta\boldsymbol{g_k}, \boldsymbol{w_k} - \boldsymbol{w_{k+1}} \rangle \le \| \eta\boldsymbol{g_k} \|_{\star} \| \boldsymbol{w_k} - \boldsymbol{w_{k+1}} \|.
    \notag
\end{equation}
Putting above two together, we have
\begin{equation}
    \begin{aligned}
    &\langle \nabla\psi(\boldsymbol{w_{k+1}}) -\nabla\psi(\boldsymbol{w_{k}}), \boldsymbol{w_{k+1}} - \boldsymbol{y} \rangle + \langle \eta\boldsymbol{g_k}, \boldsymbol{w_k} - \boldsymbol{w_{k+1}} \rangle \\
    \le& B_{\psi}
    (\boldsymbol{w_k};\boldsymbol{y}) - B_{\psi}(\boldsymbol{w_{k+1}};\boldsymbol{y}) - \frac{\lambda}{2} \| \boldsymbol{w_{k+1}} - \boldsymbol{w_k} \|^2 \\
    &+ \eta \| \boldsymbol{g_k} \|_{\star} \| \boldsymbol{w_k} - \boldsymbol{w_{k+1}} \| \\
    \le& B_{\psi}
    (\boldsymbol{w_k};\boldsymbol{y}) - B_{\psi}(\boldsymbol{w_{k+1}};\boldsymbol{y}) + \frac{\eta^2}{2\lambda} \| \boldsymbol{g_k} \|_{\star}^2,
    \end{aligned}
    \nonumber
\end{equation}
where in the last inequality, we use the fact that $-\frac{b}{2} x^2 + ax \le \frac{a^2}{2b}$ holds true for $x \in \mathbb{R}$ and $a, b > 0$.

Hence, we complete the proof that
\begin{equation}
    \begin{aligned}
    &\eta(u_k(\boldsymbol{y})-u_k(\boldsymbol{w_k})) \le  \eta \langle \boldsymbol{g_k}, \boldsymbol{y}-\boldsymbol{w_k} \rangle \\
    \le&  B_{\psi}
    (\boldsymbol{w_k};\boldsymbol{y}) - B_{\psi}(\boldsymbol{w_{k+1}};\boldsymbol{y}) - B_{\psi} (\boldsymbol{w_{k+1}};\boldsymbol{w_k}) \\
    &+ \langle \eta\boldsymbol{g_k}, \boldsymbol{w_k} - \boldsymbol{w_{k+1}} \rangle \\
    \le&  B_{\psi}(\boldsymbol{w_k};\boldsymbol{y}) - B_{\psi}(\boldsymbol{w_{k+1}};\boldsymbol{y}) + \frac{\eta^2}{2\lambda} \| \boldsymbol{g_k} \|^2_{\star}.
    \end{aligned}
    \nonumber
\end{equation}

Based on Lemma \ref{lemma:3}, we can prove the regret bound for OMD algorithm.
\begin{lemma}
\label{lemma:4}
Set $\boldsymbol{w_1} \in V$ such that $\psi$ is differentiable in $\boldsymbol{w_1}$. Assume $\eta_k$ is constant, i.e., $\eta_k=\eta, \forall k=1, 2, \cdots, K$. Then, under the assumptions of Lemma \ref{lemma:3} and $\forall \boldsymbol{y} \in V$, we can prove that
\begin{equation}
    \sum_{k=1}^K (u_k(\boldsymbol{y})-u_k(\boldsymbol{w_k})) \le \frac{B_{\psi}(\boldsymbol{w_1};\boldsymbol{y})}{\eta} + \frac{\eta}{2\lambda}\sum_{k=1}^K \| \boldsymbol{g_k} \|^2_{\star}.
    \nonumber
\end{equation}
\end{lemma}

\begin{proof}
Based on Lemma \ref{lemma:3}, we have
\begin{equation}
    \begin{aligned}
    & \sum_{k=1}^K (u_k(\boldsymbol{y}) - u_k(\boldsymbol{w_k})) \le \frac{1}{\eta} \sum_{k=1}^K (B_{\psi}(\boldsymbol{w_k};\boldsymbol{y}) - B_{\psi}(\boldsymbol{w_{k+1}};\boldsymbol{y}))\\
    &+ \frac{\eta}{2\lambda}\sum_{k=1}^K \| \boldsymbol{g_k} \|^2_{\star} 
    = \frac{1}{\eta}(B_{\psi}(\boldsymbol{w_1};\boldsymbol{y}) - B_{\psi}(\boldsymbol{w_{K+1}};\boldsymbol{y})) \\
    &+ \frac{\eta}{2\lambda}\sum_{k=1}^K \| \boldsymbol{g_k} \|^2_{\star}
    \le \frac{B_{\psi}(\boldsymbol{y};\boldsymbol{w_1})}{\eta} + \frac{\eta}{2\lambda}\sum_{k=1}^K \| \boldsymbol{g_k} \|^2_{\star}.
    \end{aligned}
    \nonumber
\end{equation}
\end{proof}

Then, we have to prove the strong convexity of $\psi(\boldsymbol{w}) = \sum_{i=1}^M w_i\ln(w_i)$ in our online learning algorithm, in order to use the conclusion of Lemma \ref{lemma:3}.

\begin{lemma}
\label{lemma:5}
$\psi(\boldsymbol{w})=\sum_{i=1}^M w_i\ln(w_i)$ is 1-strongly convex with respect to the $L_1$ norm over the set $J= \{ \boldsymbol{w} \in \mathbb{R}^M : w_i>0, \| \boldsymbol{w} \|_1=1 \}$.
\end{lemma}

\begin{proof}
We prove the strong convexity of $\psi(\boldsymbol{w})$ directly through the definition of the strong convex function, and we have
\begin{equation}
    \begin{aligned}
    \langle \nabla\psi(\boldsymbol{x}) - \nabla\psi(\boldsymbol{y}), \boldsymbol{x} - \boldsymbol{y} \rangle &= \sum_{i=1}^M (x_i - y_i)\ln(\frac{x_i}{y_i}) \\
    &= \sum_{i=1}^M (x_i\ln(\frac{x_i}{y_i}) + y_i\ln(\frac{y_i}{x_i})) \\
    &\ge \frac{1}{2} \sum_{i=1}^M |x_i-y_i|^2 + \frac{1}{2} \sum_{i=1}^M |y_i-x_i|^2 \\
    &\ge \| \boldsymbol{x} - \boldsymbol{y} \|^2_1,
    \end{aligned}
    \nonumber
\end{equation}
where we use Pinsker's inequality to complete the proof.
\end{proof}

Under the assumption of Lemma \ref{lemma:1}, when we set $\boldsymbol{w_1} = [\frac{1}{M}, \frac{1}{M}, \cdots, \frac{1}{M}]$, we have 
\begin{equation}
    B_{\psi}(\boldsymbol{w_1}; \boldsymbol{y}) = \sum_{i=1}^M y_i\ln(y_i) + \ln(M) \le \ln(M).
    \nonumber
\end{equation}
We also have $\| \boldsymbol{g_k} \|_{\star} \le 1$ for the utility functions $u_k$ in our resource allocation problem and we set $\eta = \sqrt{\frac{2\ln(M)}{K}}$. Then, from Lemma \ref{lemma:4} and Lemma \ref{lemma:5}, we have
\begin{equation}
    \begin{aligned}
        \sum_{k=1}^K (u_k(\boldsymbol{y})-u_k(\boldsymbol{w_k})) \le& \frac{B_{\psi}(\boldsymbol{w_1};\boldsymbol{y})}{\eta} + \frac{\eta}{2\lambda}\sum_{k=1}^K \| \boldsymbol{g_k} \|^2_{\star} \\
        \le& \frac{\ln(M)}{\eta} + \frac{\eta \cdot K}{2} \\
        =& \sqrt{\frac{K \cdot \ln(M)}{2}} + \frac{\sqrt{2K \cdot \ln(M)}}{2} \\
        =& \sqrt{2K\ln(M)},
    \end{aligned}
    \nonumber
\end{equation}
which completes the theoretical proof of Theorem \ref{theorem:2} for our online learning algorithm.
\end{proof}

\begin{proof}
Let $L_k := \sum_{i=1}^M w_k^i u_k^i$ denote the expected utility under the current distribution, and let $U_k(\boldsymbol{y}) := \sum_{i=1}^M y^i u_k^i$ be the utility under the fixed distribution $\boldsymbol{y}$. 
Define the cumulative regret: 
\[
\text{Regret}_K = \sum_{k=1}^K \left(U_k(\boldsymbol{y}) - L_k\right).
\]

We define a potential function based on the Kullback-Leibler (KL) divergence between a fixed target distribution $\boldsymbol{y} \in \Delta_M$ and the algorithm's current distribution $\boldsymbol{w}_k$:
\begin{equation}
    \Phi_k := D_{\text{KL}}(\boldsymbol{y} \parallel \boldsymbol{w}_k) = \sum_{i=1}^M y^i \ln \frac{y^i}{w_k^i}.
    \nonumber
\end{equation}

The multiplicative weights update is given by:
\begin{equation}
    w_{k+1}^i = \frac{w_k^i \cdot \exp(\eta u_k^i)}{Z_k}, \quad \text{where } Z_k = \sum_{j=1}^M w_k^j \cdot \exp(\eta u_k^j).
    \nonumber
\end{equation}

We analyze the change in potential $\Phi_{k+1} - \Phi_k$:
\begin{align}
    \nonumber
    \Phi_{k+1} - \Phi_k &= \sum_{i=1}^M y^i \ln \frac{w_k^i}{w_{k+1}^i} \\ \nonumber
    &= \sum_{i=1}^M y^i \left( \ln Z_k - \eta u_k^i \right) \\ \nonumber
    &= \ln Z_k - \eta \sum_{i=1}^M y^i u_k^i \\
    \nonumber
    &= \ln Z_k - \eta U_k(\boldsymbol{y}).
    \nonumber
\end{align}

Using the Taylor expansion inequality for the exponential function 
$
\exp(x) \leq 1 + x + x^2, \quad \text{for } x \in [0,1],
$ 
and assuming the learning rate $\eta$ is sufficiently small, we have that
\[
Z_k \leq \sum_{i=1}^M w_k^i \left(1 + \eta u_k^i + \eta^2 \right) = 1 + \eta \sum_{i=1}^M w_k^i u_k^i + \eta^2,
\]
where we used the fact that 
$
\sum_{i=1}^M w_k^i = 1.
$

Next, applying the inequality for the logarithm function 
$
\ln(1 + x) \leq x, \quad \text{for } x > -1,
$ 
we obtain
\[
\ln Z_k \leq \eta \sum_{i=1}^M w_k^i u_k^i + \eta^2= \eta L_k + \eta^2..
\]

Combining the above inequalities, we obtain the per-round change in KL divergence:
\begin{equation}
    \Phi_{k+1} - \Phi_k \le \eta \left( L_k - U_k(\boldsymbol{y}) \right) + \eta^2,
    \nonumber
\end{equation}
which, after rearranging, yields:
\begin{equation}
    U_k(\boldsymbol{y}) - L_k \le \frac{\Phi_k - \Phi_{k+1}}{\eta} + \eta.
    \nonumber
\end{equation}

This inequality bounds the instantaneous regret (i.e., the performance gap between the fixed comparator distribution $\boldsymbol{y}$ and the current decision $\boldsymbol{w}_k$) using the change in KL divergence and a small additive term $\eta$.

Summing both sides over all rounds $k = 1$ to $K$, we get the total regret:
\begin{align}
    \sum_{k=1}^K \left( U_k(\boldsymbol{y}) - L_k \right)
    &\le \frac{1}{\eta} \sum_{k=1}^K \left( \Phi_k - \Phi_{k+1} \right) + K \eta \nonumber \\
    &= \frac{1}{\eta} (\Phi_1 - \Phi_{K+1}) + K \eta \nonumber \\
    &\le \frac{\Phi_1}{\eta} + K \eta,
    \nonumber
\end{align}
since $\Phi_{K+1} \ge 0$ by non-negativity of KL divergence.

Now we bound the initial divergence $\Phi_1$. Since the initial weight vector is uniform, i.e., $w_1^i = \frac{1}{M}$ for all $i$, and since $\boldsymbol{y}$ lies in the probability simplex $\Delta_M$, we have:
\begin{align}
    \Phi_1 
    &= D_{\mathrm{KL}}(\boldsymbol{y} \| \boldsymbol{w}_1)
    = \sum_{i=1}^M y^i \ln \left( \frac{y^i}{1/M} \right)
    = \sum_{i=1}^M y^i \ln(M y^i)
    \le \ln M,
    \nonumber
\end{align}
where the last inequality follows from Jensen’s inequality or the fact that the KL divergence is maximized when $\boldsymbol{y}$ is a unit vector.

Therefore, we obtain the following regret bound:
\begin{equation}
    \text{Regret}_K \le \frac{\ln M}{\eta} + K \eta.
    \nonumber
\end{equation}

Optimizing the bound by choosing the learning rate $\eta = \sqrt{\frac{2 \ln M}{K}}$ yields:
\begin{equation}
    \text{Regret}_K \le \sqrt{2K \ln M},
    \nonumber
\end{equation}
which is sublinear in $K$, implying that the average regret $\frac{\text{Regret}_K}{K} \to 0$ as $K \to \infty$.
\end{proof}

\end{document}